\documentclass[12pt, draftclsnofoot, onecolumn]{IEEEtran}
\usepackage{array}
\IEEEoverridecommandlockouts
\usepackage{cite}
\usepackage{amsmath,amssymb,amsfonts}
\usepackage{algpseudocode}
\usepackage{graphicx}
\usepackage{amsmath}
\usepackage{amsthm}
\usepackage{textcomp}
\usepackage{xcolor}
\usepackage{subfigure}
\usepackage{algorithm}
\usepackage{algpseudocode}
\usepackage{makecell}
\usepackage{subfigure}
\usepackage{booktabs}
\usepackage{lipsum}
\usepackage{color}
\usepackage{makecell}
\usepackage{tabularx}
\usepackage{array}
\IEEEoverridecommandlockouts
\usepackage{cite}
\usepackage{amsmath,amssymb,amsfonts}
\usepackage{algorithm}
\usepackage{algpseudocode}
\usepackage{graphicx}
\usepackage{amsmath}
\usepackage{amsthm}
\usepackage{textcomp}
\usepackage{xcolor}
\usepackage{amsmath}
\usepackage{float} 
\usepackage{comment}
\usepackage{subfigure} 
\usepackage{color}
\usepackage{bm}
\usepackage{setspace}
\usepackage{esvect}
\usepackage[section]{placeins}
\usepackage{multirow}
\newtheorem{theorem}{\emph{\underline{Theorem}}}

\newtheorem{lemma}{\emph{\underline{Lemma}}}
\newtheorem{corollary}{\emph{\underline{Corollary}}}

\newtheorem{remark}{\bf \emph{\underline{Remark}}}
\def\phi{\varphi}

\def\({\left(}
\def\){\right)}

\def\b0{{\mathbf{0}}}

\newcommand{\diag}{\mathrm{diag}}

\def\BibTeX{{\mathrm B\kern-.05em{\sc i\kern-.025em b}\kern-.08em
    T\kern-.1667em\lower.7ex\hbox{E}\kern-.125emX}}

\begin{document}
\allowdisplaybreaks[4]
\title{
\begin{flushleft}
\author{Zhenyu\! Kang,~\IEEEmembership{\!Graduate\! Student~\!Member,~\!IEEE}, \!Changsheng~\!You,~\IEEEmembership{\!Member,~\!IEEE},\\	 and \!Rui \!Zhang,~\IEEEmembership{\!Fellow,~\!IEEE}   
\thanks{\noindent Z. Kang is with the Department of Electrical and Computer Engineering, National University of Singapore, Singapore 117583 (email: zhenyu\_kang@u.nus.edu). 
C. You is with the Department of Electronic and Electrical Engineering, Southern University of Science and Technology (SUSTech), Shenzhen 518055, China (e-mail: youcs@sustech.edu.cn).
R. Zhang is with the Department of Electrical and Computer Engineering, National University of Singapore, Singapore 117583.
He is also with the School of Science and Engineering, the Chinese University of Hong Kong (Shenzhen), China, 518172 (e-mail: elezhang@nus.edu.sg).
}
}\end{flushleft}
\huge 
$\text{Active-Passive \!IRS \!aided \!Wireless \!Communication:}$\\ New\! Hybrid\! Architecture \!and \\Elements \!Allocation\! Optimization}
\maketitle
\begin{abstract}
Intelligent reflecting surface (IRS) has emerged as a promising technology to enhance the wireless communication network coverage and capacity by dynamically controlling the radio signal propagation environment.
In contrast to the existing works that considered active or passive IRS only, we propose in this paper a new \emph{hybrid active-passive} IRS architecture that consists of both active and passive reflecting elements, thus achieving their combined advantages flexibly.
Under a practical channel setup with Rician fading where only the \emph{statistical} channel state information (CSI) is available, we study the hybrid IRS design in a multi-user communication system. Specifically, we formulate an optimization problem to maximize the achievable ergodic capacity of the worst-case user by designing the hybrid IRS beamforming and active/passive elements allocation based on the statistical CSI, subject to various practical constraints on the active-element amplification factor and amplification power consumption, as well as the total active and passive elements deployment budget.
To solve this challenging problem, we first approximate the ergodic capacity in a simpler form and then propose an efficient algorithm to solve the problem optimally. Moreover, we show that for the special case with all channels to be line-of-sight (LoS), only active elements need to be deployed when the total deployment budget is sufficiently small, while both active and passive elements should be deployed with a decreasing number ratio when the budget increases and exceeds a certain threshold.
Finally, numerical results are presented which demonstrate the performance gains of the proposed hybrid IRS architecture and its optimal design over the conventional schemes with active/passive IRS only under various practical system setups.
\end{abstract} 
\section{Introduction}
Although the fifth-generation (5G) wireless network is still under deployment, researchers have moved forward to define the next-generation or sixth-generation (6G) wireless network, with the aim for achieving more stringent performance, such as unprecedentedly high throughput, super-high reliability, ultra-low latency, extremely low power consumption, etc \cite{9040264,8766143}.
However, these targets may not be fully achieved by only relying on the existing technologies, such as massive multi-input multi-output (MIMO) and millimeter wave (mmWave) communications, which can attain enhanced performance but generally incur more substantial energy consumption and hardware cost.
On the other hand, wireless communication performance is fundamentally constrained by the wireless channel impairments such as path-loss, shadowing, and small-scale fading, which can be partially mitigated by conventional wireless communication techniques such as power control, adaptive modulation, diversity, dynamic beamforming, etc., but still remain random and uncontrolled at large. 
Recently, \emph{intelligent reflecting surface} (IRS) has emerged as a promising technology to address the above issues by leveraging massive low-cost reflecting elements to flexibly and dynamically control the radio signal propagation environment in favor of wireless communications/sensing, thus achieving substantially improved communication spectral/energy efficiency and sensing accuracy cost-effectively \cite{9326394,9724202}.

The existing works on IRS have mainly considered the wireless systems aided by \emph{passive IRS}.
Specifically, as illustrated in Fig. \ref{p-IRS}, the passive IRS is composed of a large number of passive reflecting elements with positive resistances (e.g., positive-intrinsic-negative (PIN) diodes, field-effect transistors (FETs), micro-electromechanical system (MEMS) switches) \cite{9326394}.
As such, each passive element can reflect the incident signal with a desired phase shift, while it has no signal processing/amplification capability due to the lack of transmit/receive radio frequency (RF) chains. Moreover, compared with the conventional half-duplex active relay, the passive IRS operates in a full-duplex mode and hence is free of amplification/processing noise as well as self-interference \cite{9119122}.
By properly adjusting the individual phase shifts of all passive reflecting elements with the reflection amplitude no larger than one, the reflected signal by IRS can be added constructively with that from the other propagation paths for enhancing the signal power at the intended receiver \cite{9362274,9241706} or destructively for suppressing the undesired interference \cite{9171881}.
Remarkably, it has been shown in \cite{8811733} that the passive IRS beamforming can achieve a \emph{squared power scaling order}, i.e., $\mathcal{O}(N^2)$ with $N$ denoting the number of reflecting elements, which is even higher than that of the massive MIMO with active arrays.
Extensive research has been conducted recently to efficiently incorporate passive IRS into wireless systems for various purposes, e.g., enhancing the communication throughput \cite{9714463}, reducing the outage probability \cite{9205879}, saving the transmit power \cite{8741198}, and extending the range of active relays \cite{9464248,9586067}, among others.
However, the performance gain of passive IRS is fundamentally constrained by the severe \emph{product-distance} path-loss of the reflected channel by IRS \cite{8888223}.
Two practical approaches to dealing with this problem are, respectively, deploying more passive elements at each IRS to enhance its aperture/beamforming gain and placing the passive IRSs closer to the transmitter and/or receiver for reducing the reflected channel product-distance \cite{8982186}. However, these solutions may not be suitable for practical scenarios when the space of IRS site is limited and/or its location cannot be freely selected.
\begin{figure}[t] \centering    
{\subfigure[{Passive IRS.}] {
\label{p-IRS}
\includegraphics[width=2in]{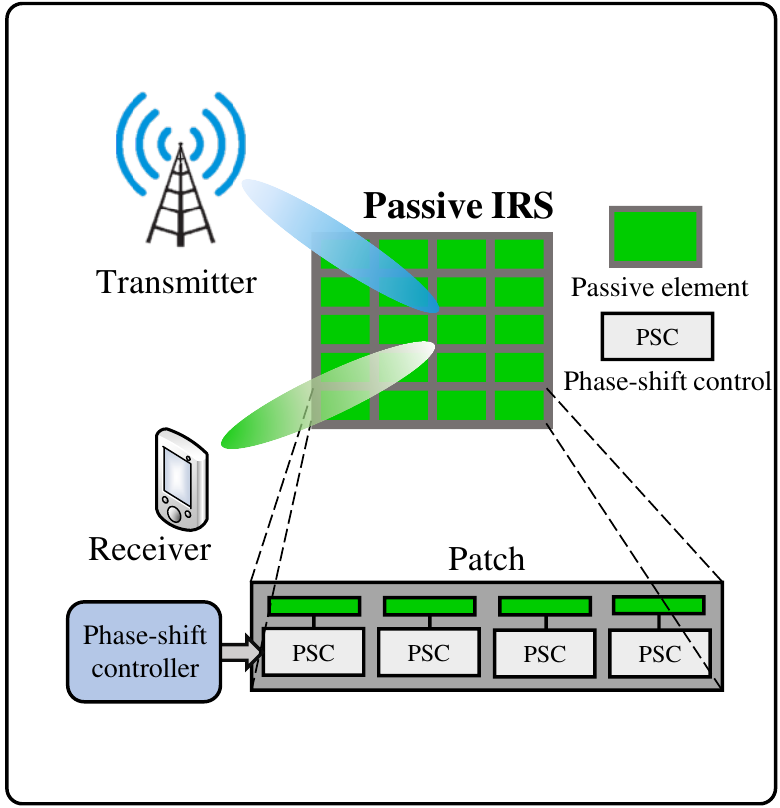}  
}}     
{\subfigure[{Active IRS.}] {\label{a-IRS}
\includegraphics[width=2in]{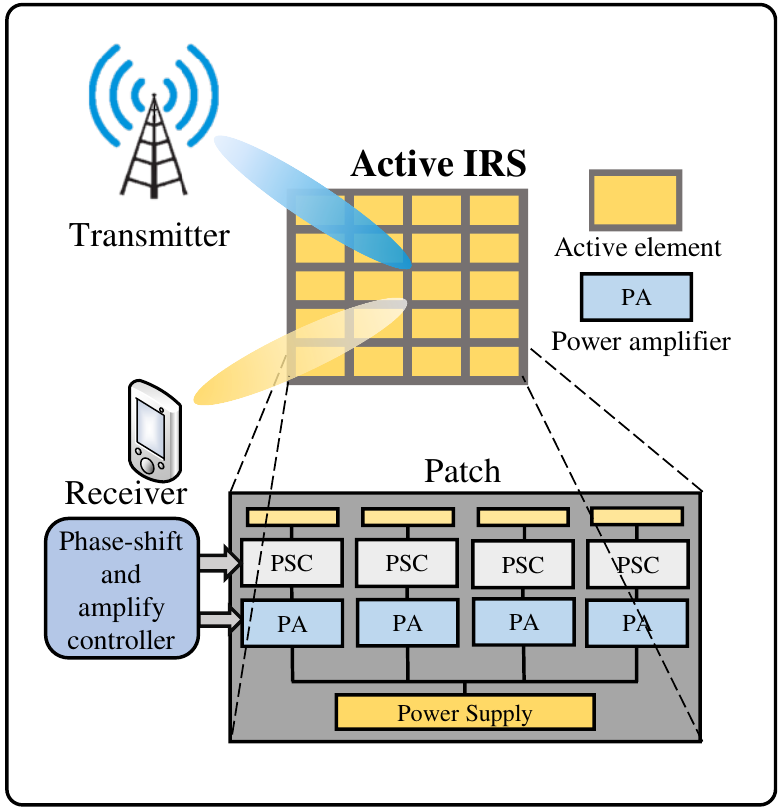}     
}}
{\subfigure[{Hybrid active-passive IRS.}] {\label{h-IRS}
\includegraphics[width=2in]{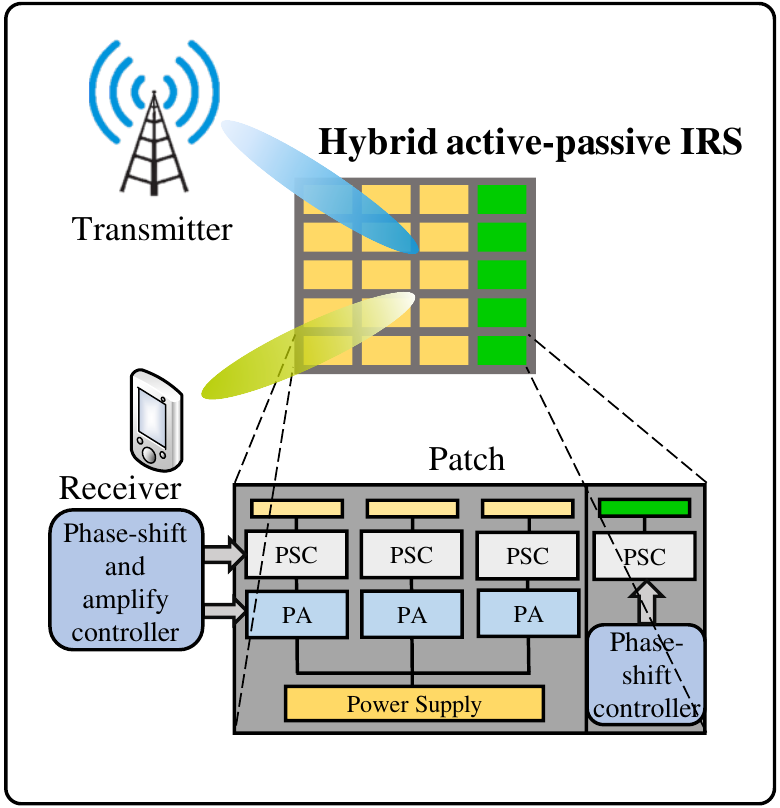}     
}}
{\caption{Different types of IRS architecture.}}
\end{figure}

To tackle the above issues, a new type of IRS, called \emph{active IRS}, has been recently proposed (see, e.g., \cite{9377648,9734027,9530750,9568854}). 
Specifically, the active IRS comprises a number of active reflecting elements, each equipped with an active load (or called negative resistance) such as the tunnel diode and negative impedance converter.
As illustrated in Fig. \ref{a-IRS}, each active load is connected to an additional power supply for signal amplification \cite{8403249,9219017}.
Therefore, the active IRS not only enables adjustable phase shifts as the passive IRS, but also allows the amplitude amplification (i.e., larger than one) of incident signals in a full-duplex mode, albeit at modestly higher hardware and energy cost than the passive IRS \cite{9377648}.
On the other hand, compared to the active relay that attaches RF chains to the antennas, the active IRS does not entail costly and power-hungry RF chain components \cite{9758764}.
The performance comparison between the active and passive IRSs has been recently studied in the literature.
For example, given the active-IRS location and power budget, it has been shown that the active IRS can achieve higher spectral efficiency \cite{9377648}, energy efficiency \cite{9568854}, and reliability \cite{9530403} than the passive IRS.
Besides, the authors in \cite{9530750} further optimized the IRS placement for both the passive- and active-IRS aided systems for rate maximization. It was shown that the passive IRS achieves higher rate performance than the active IRS with their respectively optimized placement when the number of reflecting elements is large and/or the active-element amplification power is small. 
Moreover, the authors in \cite{9734027} considered the same power budget constraint for both the passive- and active-IRS aided systems, where both the base station's (BS's) transmit power and active IRS's amplification power are considered in the case of active IRS. It was revealed that the active IRS outperforms the passive IRS only when the number of reflecting elements is small and/or the amplification power of the active IRS is sufficiently large.
\begin{figure}[t]
\centerline{\includegraphics[width=5in]{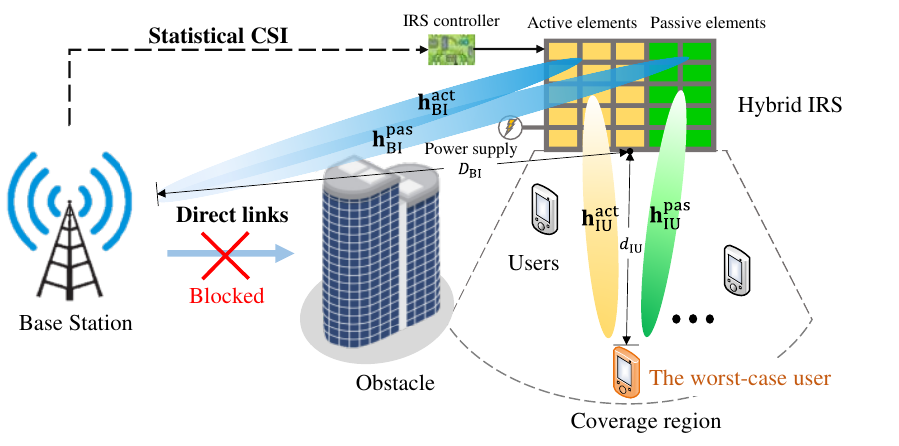}}
\caption{A hybrid active-passive IRS aided wireless communication system.}\label{sysmod}
\end{figure}
To summarize, the existing works on IRS have shown that passive and active IRSs have complementary advantages.
Specifically, the passive IRS has a higher asymptotic beamforming gain than the active IRS (i.e., $\mathcal{O}(N^2)$ versus $\mathcal{O}(N)$), thus is more appealing when the number of reflecting elements $N$ is large \cite{9530750,9377648}. In contrast, the active IRS provides additional power amplification gain, which leads to a much higher signal-to-noise ratio (SNR) than the passive IRS when $N$ is relatively small \cite{9723093,9716895,8403249}.
Besides, the active IRS generally incurs higher cost and power consumption than the passive IRS. These thus indicate that given a total budget on the IRS deployment cost (or equivalently the number of active/passive reflecting elements to be deployed), the conventional IRS architectures with either passive or active elements only in general may not achieve the optimum communication performance.

Motivated by the above, we propose in this paper a new \emph{hybrid active-passive IRS}\footnote{We use the term hybrid IRS to denote this new architecture hereafter for brevity.} architecture as shown in Fig. \ref{h-IRS} to achieve the advantages of both passive and active IRSs for further improving the performance over that with active or passive IRS alone.
Specifically, the hybrid IRS is composed of two co-located sub-surfaces, each consisting of a certain number of passive and active reflecting elements, respectively.  
In particular, we design the active versus passive reflecting elements allocation at the hybrid IRS under a given deployment budget, for optimally balancing the trade-off between the unique power amplification gain of active IRS and the higher beamforming gain of passive IRS than active IRS.
To this end, we consider a hybrid active-passive IRS aided multi-user communication system as shown in Fig. \ref{sysmod}, where a BS transmits independent data to a cluster of users.
A hybrid IRS is properly deployed at the edge of this user cluster to serve the users in its half-space reflection region over different time slots.
We consider a given deployment budget for the hybrid IRS with different costs of each active and passive reflecting element based on practical models, where an active element generally incurs higher cost than its passive counterpart.
To reduce the real-time channel estimation overhead and avoid frequent IRS reflection adjustment, we consider a practical approach that designs the hybrid IRS beamforming based on the \emph{statistical} channel state information (CSI) only, assuming the practical Rician fading channel model (i.e., only the channel path-loss parameters and Rician fading factors are assumed to be known), instead of requiring the knowledge of the instantaneous CSI of all links involved.
In the following, we summarize the main contributions of this paper.

\begin{itemize}
\item First, to guarantee the achievable rate performance of all users, we formulate an optimization problem to maximize the ergodic capacity of the worst-case user located at the boundary of the IRS reflection region.
Specifically, we assume the statistical CSI available only and jointly optimize the active/passive reflecting elements allocation, their phase shifts, and the amplification factors of active elements, subject to various practical constraints on the active-element amplification factor and amplification power consumption, as well as the total active and passive elements deployment budget.
This problem, however, is shown to be non-convex, which is difficult to be optimally solved in general.
To address this difficulty, we approximate the ergodic capacity of the worst-case user with high accuracy and thereby reformulate the original problem in a simpler form.

\item Next, we propose an efficient algorithm to solve the reformulated problem. First, we jointly optimize all elements' phase shifts and the amplification factors of active elements based on the statistical CSI, and obtain a closed-form expression for the achievable ergodic capacity with a given elements allocation.
Then, we apply the one-dimensional search to find the optimal active/passive elements allocation to maximize the ergodic capacity. 
To obtain useful insight into the optimal elements allocation, we further consider two special cases where the involved channels are line-of-sight (LoS) and follow Rayleigh fading, respectively. 
It is shown that in the former case with LoS paths, only active elements need to be deployed when the total deployment budget is sufficiently small, while both active and passive elements should be deployed with a decreasing number ratio when the budget increases and exceeds a certain threshold.

\item Last, we present extensive numerical results to evaluate the effectiveness of our proposed hybrid IRS architecture and its optimized design for rate maximization.
We show that the hybrid IRS with optimized elements allocation outperforms the conventional active/passive-only IRS architecture as well as other benchmarks. 
Moreover, the optimal active/passive elements allocation under the general Rician fading channel and that under the LoS channel are presented, which are in accordance with our theoretical analysis.
The effects of several key parameters such as the Rician fading factor, active-element amplification power, and active/passive-element deployment cost on the capacity performance and optimal active/passive elements allocation are also investigated.
\end{itemize}

It is worth noting that the authors in \cite{48550,9733238} considered a hybrid relay-reflecting intelligent surface architecture with a few relaying elements connected to power amplifiers and RF chains, which, however, significantly differs from our proposed hybrid IRS architecture with both active and passive reflecting elements. 
For example, in \cite{48550,9733238}, the relaying elements forward the signals with power-hungry RF chains, while the active reflecting elements in our proposed hybrid IRS architecture adopt the negative resistance components of much lower power consumption for signal amplification.
The remainder of this paper is organized as follows. The system model is first introduced in Section \ref{sec_sysmod}, based on which we formulate an optimization problem and approximate the worst-case user's ergodic capacity in Section \ref{sec_prob_approx}. 
In Section \ref{sec_design}, we elaborate the algorithm to solve the reformulated optimization problem, and present theoretical results that provide useful insight into the optimal IRS active/passive elements allocation.
Simulation results and pertinent discussions are presented in Section \ref{sec_sim}. Finally, the conclusions are drawn in Section \ref{sec_conclu}.

\emph{Notations}: 
Superscript $(\cdot)^H$ stands for the Hermitian transpose operation. $\mathbb{C}^{a \times b}$ denotes the space of $a \times b$ complex-valued matrices, $\mathbb{N}$ denotes the set of natural numbers, and $\mathbb{R}^+$ denotes the set of positive real numbers. 
The operation $\mathbb{E}\left\{\cdot\right\}$ returns the expected value of a random variable, $\arg(\cdot)$ returns the angle of a complex number, and $\diag{(\boldsymbol{x})}$ returns a diagonal matrix with the elements in $\boldsymbol{x}$ on its main diagonal. The notation $\jmath$ represents the imaginary unit, $\otimes$ denotes the Kronecker product, $\mathcal{C} \mathcal{N}(\mu,\sigma^2)$ denotes the circularly symmetric complex Gaussian (CSCG) random variable with mean of $\mu$ and variance of $\sigma^2$,
$[\cdot]_{m,n}$ denotes the $(m,n)$-th entry of a matrix, $[\cdot]_{m}$ denotes the $m$-th entry of a vector, and $\lfloor
x\rfloor$ denotes the largest integer that does not exceed the real number $x$.

\section{System Model}\label{sec_sysmod}
As illustrated in Fig.~\ref{sysmod}, we consider a hybrid active-passive IRS aided wireless communication system, where a single-antenna BS transmits independent data to a cluster of single-antenna users\footnote{The proposed hybrid IRS and its results obtained in this paper can be extended to the case with multi-antenna BS and multiple IRSs/user clusters, which will be investigated in our future work.}.
We assume that the BS and its served users are separated by long distance and the direct links between them are negligible due to the large channel path-loss and/or severe blockage. 
As such, a hybrid IRS comprising both active and passive reflecting elements is properly placed to serve the users in its half-space reflection region.
Moreover, we consider the time division multiple access (TDMA) scheme, where the users are served by the BS over orthogonal time slots\footnote{Under this setup, the TDMA scheme has been shown to outperform the non-orthogonal multiple access (NOMA) and orthogonal frequency division multiple access (OFDMA) schemes in terms of energy and spectral efficiency \cite{8970580}.}. 
To guarantee the system performance among all users, we consider the worst-case user performance at the boundary of the IRS reflection region with the IRS-user distance $d_{\mathrm{IU}}$ and BS-IRS distance $D_{\mathrm{BI}}$ (see Fig. \ref{sysmod}).

For ease of implementation, we assume that the hybrid IRS comprises two co-located sub-surfaces with $N_{\mathrm{pas}}$ passive and $N_{\mathrm{act}}$ active reflecting elements, respectively (see Fig. \ref{sysmod}).
Specifically, we denote by 
$\mathbf{\Psi}^{\mathrm{pas}}\triangleq \diag{(e^{\jmath\phi^{\mathrm{pas}}_1},\cdots,e^{\jmath\phi^{\mathrm{pas}}_{N_{\mathrm{pas}}}})}$
the reflection matrix of the passive sub-surface, where $\phi^{\mathrm{pas}}_n$ denotes the phase shift of the $n$-th passive element with $n\in\mathcal{N}_{\mathrm{pas}}\triangleq \{1,\cdots,N_{\mathrm{pas}}\}$, and the reflection amplitude of each passive element is set as one (i.e., its maximum value). The cost of each passive element is denoted by $W_{\mathrm{pas}}$.
On the other hand, as the active sub-surface can simultaneously amplify the signal and tune its phase shift, we denote by $\mathbf{\Psi}^{\mathrm{act}}\triangleq\mathbf{A}^{\mathrm{act}}\mathbf{\Phi}^{\mathrm{act}}$ the reflection matrix of the active sub-surface, where $\mathbf{A}^{\mathrm{act}}\triangleq\diag{(\alpha_1,\cdots,\alpha_{N_{\mathrm{act}}})}$ and $\mathbf{\Phi}^{\mathrm{act}}\triangleq
\diag{(e^{\jmath\phi^{\mathrm{act}}_1},\cdots,e^{\jmath\phi^{\mathrm{act}}_{N_{\mathrm{act}}}})}$
denote respectively its reflection amplification matrix and phase-shift matrix with $\alpha_n$ and $\phi^{\mathrm{act}}_n$ representing the amplification factor and phase shift of each active element $n\in\mathcal{N}_{\mathrm{act}}\triangleq\{1,\cdots,N_{\mathrm{act}}\}$. 
To ensure that each active reflecting element amplifies the signal, we impose a constraint on the amplification factor for each active element as $\alpha_n\geq \alpha_{\min},\forall n\in\mathcal{N}_{\mathrm{act}}$ with $\alpha_{\min}\geq 1$ \cite{9377648}. Moreover, the limited load of each active element \cite{8403249} leads to the following constraint on its maximum amplification factor: $\alpha_n\leq\alpha_{\max},\forall n\in\mathcal{N}_{\mathrm{act}}$ with $\alpha_{\max}>\alpha_{\min}$.
Let $W_{\mathrm{act}}$ represent the deployment cost of each active element, which in general is larger than that of each passive element (i.e., $W_{\mathrm{act}}>W_{\mathrm{pas}}$) due to its more sophisticated hardware (i.e., additional power amplifier and amplification control circuit \cite{9377648}) and higher static operation power (e.g., 6-20 mW for active element \cite{7920385} versus 5 mW for passive element \cite{8888223}). 
Moreover, we denote $W_0$ as the total deployment budget for the hybrid IRS such that $N_{\mathrm{act}}W_{\mathrm{act}}+N_{\mathrm{pas}}W_{\mathrm{pas}}\leq W_0$.

We assume the practical Rician fading channel model for all involved links. As such, the baseband equivalent channel from the BS to the active IRS sub-surface, denoted by $\mathbf{h}_{\mathrm{BI}}^{\mathrm{act}}\in\mathbb{C}^{N_{\mathrm{act}}\times 1}$, can be modeled as 
\begin{equation}
    \mathbf{h}_{\mathrm{BI}}^{\mathrm{act}}=\sqrt{\frac{K_{1}}{K_{1}+1}}\bar{\mathbf{h}}_{\mathrm{BI}}^{\mathrm{act}}+\sqrt{\frac{1}{K_{1}+1}}\tilde{\mathbf{h}}_{\mathrm{BI}}^{\mathrm{act}},
\end{equation}
where $K_{1}$ is the Rician fading factor of the BS$\to$IRS link, $\bar{\mathbf{h}}_{\mathrm{BI}}^{\mathrm{act}}\in\mathbb{C}^{N_{\mathrm{act}} \times 1}$ is the LoS component, and $\tilde{\mathbf{h}}_{\mathrm{BI}}^{\mathrm{act}}\in\mathbb{C}^{N_{\mathrm{act}} \times 1}$ is the non-LoS (NLoS) component.
Specifically, the LoS component can be modeled as $\bar{\mathbf{h}}_{\mathrm{BI}}^{\mathrm{act}}={h}_{\mathrm{BI}}^{\mathrm{act}}\boldsymbol{a}_{\mathrm{r}}\left(\theta_{\mathrm{BI}}^{\mathrm{r}},\vartheta_{\mathrm{BI}}^{\mathrm{r}}, N_{\mathrm{act}}\right)$, where ${h}_{\mathrm{BI}}^{\mathrm{act}}\triangleq\sqrt{\beta}/D_{\mathrm{BI}}e^{-\jmath\frac{2\pi}{\lambda}D_{\mathrm{BI}}}$ denotes the complex channel gain with $\lambda$ and $\beta$ representing respectively the carrier wavelength and the reference channel gain at a distance of 1 meter (m); and $\theta_{\mathrm{BI}}^{\mathrm{r}}\left(\vartheta_{\mathrm{BI}}^{\mathrm{r}}\right) \in[0, \pi]$ represents the azimuth (elevation) angle-of-arrival (AoA) at the IRS.
Moreover, $\boldsymbol{a}_{\mathrm{r}}\left(\theta_{\mathrm{BI}}^{\mathrm{r}}, \vartheta_{\mathrm{BI}}^{\mathrm{r}}, N_{\mathrm{act}}\right)$ denotes the receive response vector, which is given by $\boldsymbol{a}_{\mathrm{r}}\left(\theta_{\mathrm{BI}}^{\mathrm{r}}, \vartheta_{\mathrm{BI}}^{\mathrm{r}}, N_{\mathrm{act}}\right)=\boldsymbol{u}\left(\frac{2 d_{\mathrm{I}}}{\lambda} \cos \left(\theta_{\mathrm{BI}}^{\mathrm{r}}\right) \sin \left(\vartheta_{\mathrm{BI}}^{\mathrm{r}}\right), N_{\mathrm{x}}\right) \otimes$ $\boldsymbol{u}\left(\frac{2 d_{\mathrm{I}}}{\lambda} \sin \left(\theta_{\mathrm{BI}}^{\mathrm{r}}\right) \sin \left(\vartheta_{\mathrm{BI}}^{\mathrm{r}}\right), N_{\mathrm{y}}\right)$, with $\boldsymbol{u}(\varsigma, M)=[1, e^{-\jmath \pi \varsigma},$ $\ldots, $ $e^{-(M-1) \jmath \pi \varsigma}]^{T}$ representing the steering vector function, $d_{\mathrm{I}}$, $N_{\mathrm{x}}$, and $N_{\mathrm{y}}$ denoting the distance between adjacent reflecting elements and the number of active reflecting elements along the $x$- and $y$-axis on IRS surface, respectively.
Besides, the NLoS component, $\tilde{\mathbf{h}}_{\mathrm{BI}}^{\mathrm{act}}$, follows the complex Gaussian distribution with each entry $[\tilde{\mathbf{h}}_{\mathrm{BI}}^{\mathrm{act}}]_{n}\sim\frac{\sqrt{\beta}}{D_{\mathrm{BI}}}\mathcal{C} \mathcal{N}(0,1), \forall n\in\mathcal{N}_{\mathrm{act}}$.
Similarly, the baseband equivalent channel from the active IRS sub-surface to the worst-case user, $\mathbf{h}_{\mathrm{IU}}^{\mathrm{act}}\in\mathbb{C}^{N_{\mathrm{act}}\times 1}$, can be modeled as
\begin{equation}
    \mathbf{h}_{\mathrm{IU}}^{\mathrm{act}}=\sqrt{\frac{K_{2}}{K_{2}+1}}\bar{\mathbf{h}}_{\mathrm{IU}}^{\mathrm{act}}+\sqrt{\frac{1}{K_{2}+1}}\tilde{\mathbf{h}}_{\mathrm{IU}}^{\mathrm{act}},
\end{equation}
where $K_{2}$ denotes the Rician fading factor of the IRS$\to$user link, $\bar{\mathbf{h}}_{\mathrm{IU}}^{\mathrm{act}}\in\mathbb{C}^{N_{\mathrm{act}} \times 1}$ denotes the LoS component, and $\tilde{\mathbf{h}}_{\mathrm{IU}}^{\mathrm{act}}\in\mathbb{C}^{N_{\mathrm{act}} \times 1}$ denotes the NLoS component.
The baseband equivalent channel from the BS to IRS passive sub-surface, $\mathbf{h}_{\mathrm{BI}}^{\mathrm{pas}}$, and that from the IRS passive sub-surface to the worst-case user, $\mathbf{h}_{\mathrm{IU}}^{\mathrm{pas}}$, can be defined similarly as $\mathbf{h}_{\mathrm{BI}}^{\mathrm{act}}$ and $\mathbf{h}_{\mathrm{IU}}^{\mathrm{act}}$, with the details omitted for brevity. 

Based on the above, the received signal at the worst-case user aided by the hybrid IRS is given by
\begin{equation}
    y = (\mathbf{h}_{\mathrm{IU}}^{\mathrm{act}})^H\mathbf{\Psi}^{\mathrm{act}}\mathbf{h}_{\mathrm{BI}}^{\mathrm{act}}s+(\mathbf{h}_{\mathrm{IU}}^{\mathrm{pas}})^H\mathbf{\Psi}^{\mathrm{pas}}\mathbf{h}_{\mathrm{BI}}^{\mathrm{pas}}s+(\mathbf{h}_{\mathrm{IU}}^{\mathrm{act}})^H\mathbf{\Psi}^{\mathrm{act}}\mathbf{z}_{\mathrm{I}}+z_0,
\end{equation}
where $s\in\mathbb{C}$ denotes the information symbol with $\mathbb{E}\left\{\left|s\right|^{2}\right\}=P_{\mathrm{B}}$, and $P_{\mathrm{B}}$ denotes the transmit power of the BS. 
Moreover, $\mathbf{z}_{\rm{I}}\in\mathbb{C}^{N_{\mathrm{act}} \times 1}$ is the thermal noise introduced by the active elements due to signal amplification, which is assumed to follow the independent CSCG distribution, i.e., $\mathbf{z}_{\rm{I}} \sim \mathcal{C} \mathcal{N}\left(\mathbf{0}_{N_{\mathrm{act}}}, \sigma_{\mathrm{I}}^{2} \mathbf{I}_{N_{\mathrm{act}}}\right)$ with $\sigma_{\mathrm{I}}^{2}$ denoting the amplification noise power, and $z_0\sim\mathcal{C} \mathcal{N}\left(0, \sigma_0^{2}\right)$ is the additive white Gaussian noise (AWGN) at the user. 
Note that at the user's receiver, the desired signal is superposed by the reflected signals over both active and passive elements, while the noise is due to the amplification noise at active elements and the thermal noise at the receiver.

As such, the receiver SNR at the worst-case user is given by
\begin{equation}
    \gamma = \frac{P_{\mathrm{B}}|(\mathbf{h}_{\mathrm{IU}}^{\mathrm{act}})^H\mathbf{\Psi}^{\mathrm{act}}\mathbf{h}_{\mathrm{BI}}^{\mathrm{act}}+(\mathbf{h}_{\mathrm{IU}}^{\mathrm{pas}})^H\mathbf{\Psi}^{\mathrm{pas}}\mathbf{h}_{\mathrm{BI}}^{\mathrm{pas}}|^2}{\sigma_{\mathrm{I}}^2\|(\mathbf{h}_{\mathrm{IU}}^{\mathrm{act}})^H\mathbf{\Psi}^{\mathrm{act}}\|^2+\sigma_0^2}.\label{snr_hybrid}
\end{equation}
Thus, the ergodic capacity achieved by the worst-case user in the hybrid IRS aided wireless communication system is given by
\begin{equation}
    C=\mathbb{E}\left\{\log _{2}(1+\gamma)\right\},\label{ergo_capa}
\end{equation}
where the expectation is taken over the random NLoS components in all channels involved.
\section{Problem Formulation and Ergodic Capacity Analysis}\label{sec_prob_approx}
We aim to maximize the ergodic capacity of the worst-case user subject to a total deployment budget of $W_0$ by optimizing the numbers of active and passive elements,
$N_{\mathrm{act}}$ and $N_{\mathrm{pas}}$, the IRS phase shifts, \{$\mathbf{\Phi}^{\mathrm{act}}$, $\mathbf{\Psi}^{\mathrm{pas}}$\}, and the active-element amplification matrix, $\mathbf{A}^{\mathrm{act}}$. This problem can be formulated as follows.
\begin{align}
    &\hspace{-0.2cm}\mathrm{(P1)}~~~~~\max_{\mathbf{\Phi}^{\mathrm{act}},\mathbf{\Psi}^{\mathrm{pas}},\mathbf{A}^{\mathrm{act}},N_{\mathrm{act}},N_{\mathrm{pas}}}
    \quad~~\mathbb{E}\left\{\log _{2}\left(1\!+\!\frac{P_{\mathrm{B}}|(\mathbf{h}_{\mathrm{IU}}^{\mathrm{act}})^H\mathbf{\Psi}^{\mathrm{act}}\mathbf{h}_{\mathrm{BI}}^{\mathrm{act}}\!+\!(\mathbf{h}_{\mathrm{IU}}^{\mathrm{pas}})^H\mathbf{\Psi}^{\mathrm{pas}}\mathbf{h}_{\mathrm{BI}}^{\mathrm{pas}}|^2}{\sigma_{\mathrm{I}}^2\|(\mathbf{h}_{\mathrm{IU}}^{\mathrm{act}})^H\mathbf{\Psi}^{\mathrm{act}}\|^2\!+\!\sigma_0^2}\right)\right\} \label{obj_func_orig}\\
    &\qquad\qquad~~~~\!~\quad~\text{s.t.} \qquad\qquad\qquad~~0<\phi_n^{\mathrm{act}}\leq 2\pi,\forall n\in\mathcal{N}_{\mathrm{act}},\label{cons_phase_act}\\
    &\qquad\qquad~~~~~~~~\quad~\qquad\qquad\qquad~~0<\phi_n^{\mathrm{pas}}\leq 2\pi,\forall n\in\mathcal{N}_{\mathrm{pas}},\\
    &\qquad\qquad~~~~~~~~\quad~\qquad\qquad\qquad~~\alpha_{\min}\leq \alpha_n\leq \alpha_{\max},\forall n\in\mathcal{N}_{\mathrm{act}},\label{cons_alpha}\\
    &\qquad\qquad~~~~~~~~\quad~\qquad\qquad\qquad~~ \mathbb{E}\left\{P_{\mathrm{B}}\|\mathbf{\Psi}^{\mathrm{act}}\mathbf{h}_{\mathrm{BI}}^{\mathrm{act}}\|^2+\sigma_{\mathrm{I}}^2\|\mathbf{\Psi}^{\mathrm{act}}\|^2\right\}\leq P_{\mathrm{I}},\label{cstr_power_HI}\\
    &\qquad\qquad~~~~~~~~\quad~\qquad\qquad\qquad~~N_{\mathrm{act}}W_{\mathrm{act}}+N_{\mathrm{pas}}W_{\mathrm{pas}}\leq W_0,\label{cons_C}\\
    &\qquad\qquad~~~~~~~~\quad~\qquad\qquad\qquad~~N_{\mathrm{act}}\in\mathbb{N},N_{\mathrm{pas}}\in\mathbb{N}\label{cons_C_AnP},
\end{align}
where the constraint \eqref{cstr_power_HI} indicates that the average amplification power of all active elements over the Rician fading channels is constrained by the average amplification power budget, $P_{\mathrm{I}}$.
Note that different from the existing works on IRS-aided wireless communications that generally require the instantaneous CSI knowledge (e.g., \cite{9133142}), we consider a practical scenario where only the statistical CSI, i.e., $\{\bar{\mathbf{h}}_{\mathrm{IU}}^{\mathrm{act}},\bar{\mathbf{h}}_{\mathrm{BI}}^{\mathrm{act}},$ $\bar{\mathbf{h}}_{\mathrm{IU}}^{\mathrm{pas}},$ $\bar{\mathbf{h}}_{\mathrm{BI}}^{\mathrm{pas}},$ $K_1,$ $K_2\}$ is known \textit{a priori}, which suffices for the design of active/passive elements allocation for maximizing the ergodic capacity (in the worst case).
This approach also reduces the real-time channel estimation overhead and avoids frequent IRS reflection adjustment for each user. 

For problem (P1), note that it includes the conventional IRS architectures with all-passive and all-active elements as
two special cases. Specifically, when $N_{\mathrm{act}}=0$, the hybrid IRS reduces to the conventional passive IRS and we have $\mathbf{\Psi}^{\mathrm{act}}=\mathbf{0}_{N_{\mathrm{act}}\times N_{\mathrm{act}}}$. On the other hand, when $N_{\mathrm{pas}}=0$, it reduces to the conventional active IRS and thus $\mathbf{\Psi}^{\mathrm{pas}}=\mathbf{0}_{N_{\mathrm{pas}}\times N_{\mathrm{pas}}}$.
However, problem (P1) is challenging to solve even in the above two special cases, since the phase shifts are coupled with the amplification factors in the function of the ergodic capacity (see \eqref{snr_hybrid} and \eqref{ergo_capa}). Moreover, the numbers of active and passive elements are discrete, rendering the design objective a complicated function and the constraints in \eqref{cstr_power_HI}--\eqref{cons_C_AnP} non-convex.

To address the above issues, we first analyze the ergodic capacity to approximate the objective function of problem (P1) in a simpler form.

{\color{black}\begin{lemma}\label{lem_C_approx}
\textbf{\emph{(Ergodic Capacity Approximation)}} \emph{The ergodic capacity in \eqref{obj_func_orig} can be approximated by}
\begin{align}
    C\approx \tilde{C}\triangleq\log _{2}\left(1+\frac{x_{\mathrm{L}}+x_{\mathrm{NL,act}}+x_{\mathrm{NL,pas}}}{z_{\mathrm{L,act}}+z_{\mathrm{NL,act}}+\sigma_0^2}\right),\label{sig_approx}
\end{align}
\emph{where}
\begin{align}
    &x_{\mathrm{L}} \triangleq \frac{K_1K_2P_{\mathrm{B}}}{(K_1\!+\!1)(K_2\!+\!1)}\left|(\bar{\mathbf{h}}_{\mathrm{IU}}^{\mathrm{act}})^H\mathbf{\Psi}^{\mathrm{act}}\bar{\mathbf{h}}_{\mathrm{BI}}^{\mathrm{act}}\!+\!(\bar{\mathbf{h}}_{\mathrm{IU}}^{\mathrm{pas}})^H\mathbf{\Psi}^{\mathrm{pas}}\bar{\mathbf{h}}_{\mathrm{BI}}^{\mathrm{pas}}\right|^2,\label{x_1}\\
    &x_{\mathrm{NL,act}}\triangleq\frac{P_{\mathrm{B}}}{(K_1\!+\!1)(K_2\!+\!1)}\Big(\frac{K_1\beta}{d_{\mathrm{IU}}^2}\|\mathbf{\Psi}^{\mathrm{act}}\bar{\mathbf{h}}_{\mathrm{BI}}^{\mathrm{act}}\|^2\!+\!\frac{K_2\beta}{D_{\mathrm{BI}}^2}\|(\bar{\mathbf{h}}_{\mathrm{IU}}^{\mathrm{act}})^H\mathbf{\Psi}^{\mathrm{act}}\|^2\!+\!\frac{\beta^2}{D_{\mathrm{BI}}^2d_{\mathrm{IU}}^2}\sum_{n=1}^{N_{\mathrm{act}}}\alpha^2_{n}\Big),\label{x_2}\\
    &x_{\mathrm{NL,pas}}\triangleq\frac{P_{\mathrm{B}}}{(K_1\!+\!1)(K_2\!+\!1)}\Big(\frac{K_1\beta}{d_{\mathrm{IU}}^2}\|\mathbf{\Psi}^{\mathrm{pas}}\bar{\mathbf{h}}_{\mathrm{BI}}^{\mathrm{pas}}\|^2\!+\!\frac{K_2\beta}{D_{\mathrm{BI}}^2}\|(\bar{\mathbf{h}}_{\mathrm{IU}}^{\mathrm{pas}})^H\mathbf{\Psi}^{\mathrm{pas}}\|^2\!+\!\frac{\beta^2N_{\mathrm{pas}}}{D_{\mathrm{BI}}^2d_{\mathrm{IU}}^2}\Big),\label{x_3}\\
    &{z_{\mathrm{L,act}}}\triangleq{\frac{K_2\sigma_{\rm I}^2}{K_2\!+\!1}\|(\bar{\mathbf{h}}_{\mathrm{IU}}^{\mathrm{act}})^H\mathbf{\Psi}^{\mathrm{act}}\|^2},\quad
    {z_{\mathrm{NL,act}}}\triangleq{\frac{\sigma_{\rm I}^2}{K_2\!+\!1}\|(\tilde{\mathbf{h}}_{\mathrm{IU}}^{\mathrm{act}})^H\mathbf{\Psi}^{\mathrm{act}}\|^2}.
 \end{align}
\end{lemma}}
\begin{proof}
See Appendix \ref{proof_lem1}.
\end{proof}
The accuracy of the approximation in Lemma \ref{lem_C_approx} will be numerically verified in Section \ref{sec_sim}.
Similarly, the average amplification power consumption can be expressed as
\begin{align}
    \!\!\!\!\!\mathbb{E}\left\{P_{\mathrm{B}}\|\mathbf{\Psi}^{\mathrm{act}}\mathbf{h}_{\mathrm{BI}}^{\mathrm{act}}\|^2\!+\!\sigma_{\mathrm{I}}^2\|\mathbf{\Psi}^{\mathrm{act}}\|^2\right\}\!=
    \!\frac{P_{\mathrm{B}}}{K_2\!+\!1}\left(K_2\|\mathbf{\Psi}^{\mathrm{act}}\bar{\mathbf{h}}_{\mathrm{BI}}^{\mathrm{act}}\|^2\!+\!\mathbb{E}\left\{\|\mathbf{\Psi}^{\mathrm{act}}\tilde{\mathbf{h}}_{\mathrm{BI}}^{\mathrm{act}}\|^2\right\}\right)\!+\!\sum_{n=1}^{N_{\mathrm{act}}}\sigma_{\mathrm{I}}^2\alpha^2_{n}\label{power_decomp}.
\end{align}
As such, problem (P1) is reformatted as follows by using Lemma \ref{lem_C_approx} and substituting \eqref{power_decomp} into \eqref{cstr_power_HI},
\begin{align}
    &\mathrm{(P2)}\max_{\mathbf{\Phi}^{\mathrm{act}},\mathbf{\Psi}^{\mathrm{pas}},\mathbf{A}^{\mathrm{act}},N_{\mathrm{act}},N_{\mathrm{pas}}}
    \quad~~\log _{2}\left(1+\frac{x_{\mathrm{L}}+x_{\mathrm{NL,act}}+x_{\mathrm{NL,pas}}}{z_{\mathrm{L,act}}+z_{\mathrm{NL,act}}+\sigma_0^2}\right) \\
    &\qquad\qquad~~~~\!~\text{s.t.} \qquad\qquad~~\eqref{cons_phase_act}-\eqref{cons_alpha}, \eqref{cons_C},\eqref{cons_C_AnP}\nonumber,\\
    &\qquad~~~~~~~\qquad\qquad\qquad~~ \frac{P_{\mathrm{B}}}{K_2\!+\!1}\left(K_2\|\mathbf{\Psi}^{\mathrm{act}}\bar{\mathbf{h}}_{\mathrm{BI}}^{\mathrm{act}}\|^2\!+\!\mathbb{E}\left\{\|\mathbf{\Psi}^{\mathrm{act}}\tilde{\mathbf{h}}_{\mathrm{BI}}^{\mathrm{act}}\|^2\right\}\right)\!+\!\sum_{n=1}^{N_{\mathrm{act}}}\sigma_{\mathrm{I}}^2\alpha^2_{n}\leq\! P_{\mathrm{I}}.\label{cstr_power_HI_decomp}
\end{align}

\section{Optimal Solution to Problem (P2)}\label{sec_design}
In this section, we aim to optimally solve problem (P2) and gain useful insight into the optimal active/passive elements allocation at the hybrid IRS.
Specifically, we first obtain the optimal hybrid IRS beamforming and active/passive elements allocation based on the statistical CSI under the general Rician fading channel model. Then, we consider two special channel setups, i.e., the LoS and Rayleigh fading channel models, to gain useful insight into the optimal active/passive elements allocation.

\subsection{IRS Phase Shift Optimization}
Given any feasible active/passive elements allocation and active-element amplification factors, it can be easily shown that the approximated ergodic capacity in \eqref{sig_approx} is maximized when the LoS components of the IRS-associated channels, i.e., $(\bar{\mathbf{h}}_{\mathrm{IU}}^{\mathrm{act}})^H\mathbf{\Psi}^{\mathrm{act}}\bar{\mathbf{h}}_{\mathrm{BI}}^{\mathrm{act}}$ and $(\bar{\mathbf{h}}_{\mathrm{IU}}^{\mathrm{pas}})^H\mathbf{\Psi}^{\mathrm{pas}}\bar{\mathbf{h}}_{\mathrm{BI}}^{\mathrm{pas}}$, are phase-aligned. The optimal IRS phase shifts are thus given by
\begin{align}
    &\phi_{n}^{\mathrm{act}} = \arg([\bar{\mathbf{h}}_{\mathrm{IU}}^{\mathrm{act}}]_n)-\arg([\bar{\mathbf{h}}_{\mathrm{BI}}^{\mathrm{act}}]_n),\forall n\in\mathcal{N}_{\mathrm{act}},\label{opt_phase_1}\\
    &\phi_{n}^{\mathrm{pas}} = \arg([\bar{\mathbf{h}}_{\mathrm{IU}}^{\mathrm{pas}}]_n)-\arg([\bar{\mathbf{h}}_{\mathrm{BI}}^{\mathrm{pas}}]_n),\forall n\in \mathcal{N}_{\mathrm{pas}},\label{opt_phase_2}
\end{align}
where the optimal phase shifts of the active and passive sub-surfaces have the similar form.
Then, by substituting the optimal phase shifts in \eqref{opt_phase_1} and \eqref{opt_phase_2} into \eqref{sig_approx}, we have
\begin{align}
     &x_{\mathrm{L}}=\gamma_1\left(\sum_{n=1}^{N_{\mathrm{act}}}\alpha_n+N_{\mathrm{pas}}\right)^2P_{\mathrm{B}}\beta^2/D^2_{\mathrm{BI}}d^2_{\mathrm{IU}},\\
     &x_{\mathrm{NL,act}}+x_{\mathrm{NL,pas}}=\gamma_2\left(\sum_{n=1}^{N_{\mathrm{act}}}\alpha_n^2+N_{\mathrm{pas}}\right)P_{\mathrm{B}}\beta^2(K_1+K_2+1)/D^2_{\mathrm{BI}}d^2_{\mathrm{IU}},\label{x_NLoS}\\
     &z_{\mathrm{L,act}}+z_{\mathrm{NL,act}}=\sum_{n=1}^{N_{\mathrm{act}}}\alpha_n^2\sigma_{\mathrm{I}}^2\beta/d_{\mathrm{IU}}^{2},\label{n_act}
\end{align}
where 
\begin{equation}
    \gamma_1 \triangleq \frac{K_1K_2}{(K_1+1)(K_2+1)}, \gamma_2 \triangleq \frac{1}{(K_1+1)(K_2+1)}.\label{gamma_approx}
\end{equation}
Note that the considered hybrid IRS beamforming with the statistical CSI only is designed based on the LoS components while ignoring the NLoS components of the involved channels.

\subsection{Amplification Factor Optimization for Active Elements}

In this subsection, we optimize the amplification factor given any feasible active/passive elements allocation and optimal IRS phase shifts (see \eqref{opt_phase_1} and \eqref{opt_phase_2}). To this end, we first present an important lemma as follows.
\begin{lemma}\label{lem_min}
\textbf{\emph{(Minimum amplification power for active elements)}}
\emph{Given the total deployment budget $W_0$, the hybrid IRS should employ passive elements only (i.e., $N_{\mathrm{act}}=0$) when}
\begin{equation}
    P_{\mathrm{I}} < \alpha^2_{\min}\left(P_{\mathrm{B}}\beta/D^2_{\mathrm{BI}}+\sigma^2_{\mathrm{I}}\right),\label{lem_pas}
\end{equation}
\emph{and employ active elements (i.e., $N_{\mathrm{act}}>0$) otherwise.}
\end{lemma}
\begin{proof}
{It can be shown that when \eqref{lem_pas} holds, the maximum amplification factor that $P_{\mathrm{I}}$ can support is smaller than its allowable minimum value, i.e., $\alpha_{n}<\alpha_{\min},\forall n\in \mathcal{N}_{\mathrm{act}}$, thus making it infeasible to satisfy $\alpha_{n}\geq\alpha_{\min}$.}
\end{proof}

Next, we consider the non-trivial case with $P_{\mathrm{I}} \geq \alpha^2_{\min}\left(P_{\mathrm{B}}\beta/D^2_{\mathrm{BI}}+\sigma^2_{\mathrm{I}}\right)$ and hence $N_{\mathrm{act}}>0$. Given the optimal phase shift design in \eqref{opt_phase_1} and \eqref{opt_phase_2}, problem (P2) reduces to the following problem for optimizing the amplification factors of the $N_{\mathrm{act}}$ active elements.
\begin{align}
    &\mathrm{(P3)}~~~~~\max_{\mathbf{A}^{\mathrm{act}}}
    \quad~~\log _{2}\left(1+\frac{x_{\mathrm{L}}+x_{\mathrm{NL,act}}+x_{\mathrm{NL,pas}}}{z_{\mathrm{L,act}}+z_{\mathrm{NL,act}}+\sigma_0^2}\right) \\
    &\qquad\!~~~~\quad~\text{s.t.} \qquad~~\eqref{cons_alpha}\nonumber,\\
    &\qquad~\qquad\qquad\qquad\sum_{n=1}^{N_{\mathrm{act}}}\alpha^2_{n}\left(P_{\mathrm{B}}\beta/D^2_{\mathrm{BI}}+\sigma^2_{\mathrm{I}}\right)\leq P_{\mathrm{I}}.\label{cstr_power_HI_decomp_2}
\end{align}
The constraints in \eqref{cons_alpha} and \eqref{cstr_power_HI_decomp_2} indicate that the amplification factors are bounded by both the hardware limitation and the power budget for all active elements. 
To simplify the analysis, we have the following lemma.
\begin{lemma}\label{lem_P_I_range}
\textbf{\emph{(Favorable amplification power condition)}}
\emph{For problem (P3), the constraint in \eqref{cons_alpha} is always satisfied if}
\begin{equation}
\alpha^2_{\min}\left(P_{\mathrm{B}}\beta/D^2_{\mathrm{BI}}+\sigma^2_{\mathrm{I}}\right)\leq P_{\mathrm{I}}\leq W_0\alpha^2_{\max}\left(P_{\mathrm{B}}\beta/D^2_{\mathrm{BI}}+\sigma^2_{\mathrm{I}}\right)/W_{\mathrm{act}}.\label{cons_pb}
\end{equation}
\end{lemma}
\begin{proof}
First, it can be shown that when 
\begin{equation}
    P_{\mathrm{I}}\geq W_0\alpha^2_{\max}\left(P_{\mathrm{B}}\beta/D^2_{\mathrm{BI}}+\sigma^2_{\mathrm{I}}\right)/W_{\mathrm{act}},\label{lem_max}
\end{equation}
the amplification factors of $N_{\mathrm{act}}$ active elements are equal to their maximum value, i.e., $\alpha_{n}=\alpha_{\max},\forall n\in\mathcal{N}_{\mathrm{act}}$, where the power amplification is constrained by the maximum load.
Then, combining \eqref{lem_max} and Lemma \ref{lem_min}, there always exists a feasible $\alpha_{n}\in[\alpha_{\min},\alpha_{\max}]$ such that the constraint in \eqref{cons_alpha} of problem (P3) is satisfied if \eqref{cons_pb} holds by choosing an appropriate $N_{\mathrm{act}}\in[0,W_0/W_{\mathrm{act}}]$.
\end{proof}

Lemma \ref{lem_P_I_range} gives the favorable amplification power condition for active elements that are able to operate in the amplification mode (i.e., $\alpha_{n}\geq\alpha_{\min}$) and make full use of the amplification power (i.e., $\alpha_{n}\leq\alpha_{\max}$).
In contrast, if the amplification power is too large, i.e., satisfying \eqref{lem_max}, the amplification factor of each active element is constrained by its hardware limitation, i.e., $\alpha_n=\alpha_{\max},\forall n\in\mathcal{N}_{\mathrm{act}}$. On the other hand, if the amplification power is insufficient, i.e., when \eqref{lem_pas} holds, the active elements cannot operate in the amplification mode, i.e., $\alpha_n< \alpha_{\min},\exists n\in \mathcal{N}_{\mathrm{act}}$, thus making problem (P3) infeasible. 
Therefore, in the sequel, we only consider the case where $P_{\mathrm{I}}$ satisfies the favorable amplification power condition given in \eqref{cons_pb} in order to draw useful insight into the optimal IRS active/passive elements allocation.
\begin{lemma}\label{lem1}
\textbf{\emph{(Optimal amplification factors)}}
\emph{Given \eqref{cons_pb} and $N_{\mathrm{act}}>0$, the optimal amplification factor for each active element in problem (P3) is given by}
\begin{align}
    \alpha_{n}=\alpha^* \triangleq \sqrt{\frac{P_{\mathrm{I}}/N_{\mathrm{act}}}{P_{\mathrm{B}}\beta/D^2_{\mathrm{BI}}+\sigma^2_{\mathrm{I}}}}, \forall n \in\mathcal{N}_{\mathrm{act}}.\label{opt_a_n}
\end{align}
\end{lemma}
\begin{proof}
See Appendix \ref{proof_lem2}.
\end{proof}
Lemma \ref{lem1} shows that all active reflecting elements should adopt a common amplification factor due to the same path-loss over the LoS channels associated with the active elements.
Moreover, it is observed from \eqref{opt_a_n} that given the amplification power $P_{\mathrm{I}}$, the optimal amplification factor, $\alpha^*$, decreases with the number of active elements or signal power.
\begin{remark}
\textbf{\emph{(Amplification noise power)}}
\emph{By substituting $\alpha^*$ in \eqref{opt_a_n} into \eqref{n_act}, the amplification noise power is given by}
\begin{equation}
    z_{\mathrm{L,act}}+z_{\mathrm{NL,act}}=\frac{\mathcal{I}_{\mathbb{R}^+}(N_{\mathrm{act}})P_{\mathrm{I}}\sigma_{\mathrm{I}}^2\beta/d_{\mathrm{IU}}^{2}}{P_{\mathrm{B}}\beta/D_{\mathrm{BI}}^{2}+\sigma_{\mathrm{I}}^{2}},\label{noise_term}
\end{equation}
\emph{where the indicator function $\mathcal{I}_{\mathbb{R}^+}(N_{\mathrm{act}})=1$ when $N_{\mathrm{act}}>0$ and $\mathcal{I}_{\mathbb{R}^+}(N_{\mathrm{act}})=0$ otherwise, i.e., the noise term exists only when the hybrid IRS consists of active elements.
Note that the amplification noise power in \eqref{noise_term} is a constant that depends on the total amplification power, $P_{\mathrm{I}}$, and the path-loss of the BS-IRS and IRS-user channels. 
}
\end{remark}
Based on Lemma \ref{lem1}, the ergodic capacity of the worst-case user with the optimal IRS phase shifts in \eqref{opt_phase_1} and \eqref{opt_phase_2} and amplification factors of active elements in \eqref{opt_a_n} is given by
\begin{equation}
\tilde C^*\!=\!\log _{2}\left(1+\frac{\frac{P_{\mathrm{B}} \beta^{2}}{D^2_{\mathrm{BI}}d^2_{\mathrm{IU}}}\left(\gamma_{1} \left(\sqrt{A_{\mathrm{sum}}N_{\mathrm{act}}}+N_{\mathrm{pas}}\right)^{2}+\gamma_{2}\left(A_{\mathrm{sum}}\!+\!N_{\mathrm{pas}}\right)\right)}{A_{\mathrm{sum}} \sigma_{\mathrm{I}}^{2}\beta / d_{\mathrm{IU}}^{2}\!+\!\sigma_{0}^{2}}\right),\label{capa_approx2}
\end{equation}
where $A_{\mathrm{sum}}\triangleq\frac{\mathcal{I}_{\mathbb{R}^+}(N_{\mathrm{act}})P_{\mathrm{I}}}{P_{\mathrm{B}}\beta/D_{\mathrm{BI}}^{2}+\sigma_{\mathrm{I}}^{2}}$.

\subsection{Active/Passive Elements Allocation Optimization}\label{ele_alloc}
Given the optimal IRS phase shifts and amplification factors of active elements, in the next, we focus on optimizing the active/passive elements allocation for maximizing the ergodic capacity subject to the total deployment budget constraint, which is formulated as follows.
\begin{align}
    &\mathrm{(P4)}~~~~~\max_{N_{\mathrm{act}},N_{\mathrm{pas}}}
    \quad~~\log _{2}\left(1+\frac{\frac{P_{\mathrm{B}} \beta^{2}}{D^2_{\mathrm{BI}}d^2_{\mathrm{IU}}}\left(\gamma_{1} \left(\sqrt{A_{\mathrm{sum}}N_{\mathrm{act}}}+N_{\mathrm{pas}}\right)^{2}+\gamma_{2}\left(A_{\mathrm{sum}}\!+\!N_{\mathrm{pas}}\right)\right)}{A_{\mathrm{sum}} \sigma_{\mathrm{I}}^{2}\beta / d_{\mathrm{IU}}^{2}\!+\!\sigma_{0}^{2}}\right) \\
    &\qquad~~~\!~\quad~\text{s.t.}\qquad\quad \eqref{cons_C},\eqref{cons_C_AnP}\nonumber.
\end{align}

Problem (P4) is a non-convex optimization problem due to the discrete active/passive-element number and the non-concave objective function, making it difficult to obtain a closed-form expression for the optimal solution in general.
Although the optimal number of active (passive) elements can be numerically obtained by applying the one-dimensional search over $N_{\mathrm{act}}=\{0,1,\cdots,\lfloor\frac{W_0}{W_{\mathrm{act}}}\rfloor\}$ and hence $N_{\mathrm{pas}} = \lfloor\frac{W_0-N_{\mathrm{act}}W_{\mathrm{act}}}{W_{\mathrm{pas}}}\rfloor$, it yields little useful insight into the optimal IRS active/passive elements allocation. 
Thus, we consider two special channel setups in the following to obtain closed-form expressions for their corresponding optimal active/passive elements allocation.

\subsubsection{LoS channel model}
For the LoS channels with $K_1\to\infty$ and $K_2\to\infty$, we obtain from \eqref{gamma_approx} that $\gamma_1\to 1$ and $\gamma_2\to 0$. Then, the approximated ergodic capacity, $\tilde C$ in \eqref{sig_approx}, is equal to the exact capacity, $C$ in \eqref{ergo_capa}.
Note that the term $A_{\mathrm{sum}}$ is a positive constant, i.e., $A_{\mathrm{sum}}=\frac{P_{\mathrm{I}}}{P_{\mathrm{B}}\beta / D_{\mathrm{BI}}^{2}+\sigma_{\mathrm{I}}^{2}}$ when $N_{\mathrm{act}}>0$, and $A_{\mathrm{sum}}=0$ when $N_{\mathrm{act}}=0$.
Therefore, in the following, we solve problem (P4) in two cases, corresponding to the case of passive IRS with $N_{\mathrm{act}}=0$ and the case of hybrid IRS with $N_{\mathrm{act}}>0$ and $N_{\mathrm{pas}}>0$, respectively. 
To address the discrete active/passive elements deployment cost in constraints \eqref{cons_C} and \eqref{cons_C_AnP}, we first relax the integer values $N_{\mathrm{act}}$ and $N_{\mathrm{pas}}$ into continuous values, $\tilde N_{\mathrm{act}}$ and $\tilde N_{\mathrm{pas}}$, and then apply the integer rounding technique to reconstruct them based on their optimal continuous values.
Moreover, it can be shown that the equality in \eqref{cons_C} holds in the optimal solution to problem (P4), i.e., $\tilde N_{\mathrm{pas}}=\frac{W_0-W_{\mathrm{act}}\tilde N_{\mathrm{act}}}{W_{\mathrm{pas}}}$.

First, consider the case of $\tilde N_{\mathrm{act}}=0$. By substituting $\tilde N_{\mathrm{pas}}=\frac{W_0}{W_{\mathrm{pas}}}$, $\gamma_1\to 1$, and $\gamma_2\to 0$ into \eqref{capa_approx2}, the hybrid IRS reduces to the passive IRS for which the achievable capacity under the LoS channels is given by
\begin{equation}
   C = C^*_{\text {L,pas}} \triangleq  \log _{2}\left(1+\frac{W_0^{2} P_{\mathrm{B}} \beta^{2}}{W_{\mathrm{pas}}^{2} D_{\mathrm{BI}}^{2} d_{\mathrm{IU}}^{2} \sigma_{0}^{2}}\right).\label{C_p}
\end{equation} 
Next, when $\tilde N_{\mathrm{act}}>0$ under the LoS channel model, the achievable capacity in \eqref{capa_approx2} reduces to
\begin{equation}
    C = C_{\text {L,hyb}} \triangleq \log _{2}\left(1+\frac{P_{\mathrm{B}}\beta^2(\sqrt{A_{\mathrm{sum}}\tilde N_{\mathrm{act}}}+\tilde N_{\mathrm{pas}})^2/D_{\mathrm{BI}}^2d_{\mathrm{IU}}^2}{A_{\mathrm{sum}}\sigma_{\mathrm{I}}^2\beta/d_{\mathrm{IU}}^2+\sigma_0^2}\right).\label{C_h}
\end{equation}
By substituting $\tilde N_{\mathrm{pas}}=\frac{W_0-W_{\mathrm{act}}\tilde N_{\mathrm{act}}}{W_{\mathrm{pas}}}$ into \eqref{C_h}, the optimal solution to problem (P4) given $\tilde N_{\mathrm{act}}>0$ can be obtained by solving the following equivalent maximization problem.
\begin{align}
    &\mathrm{(P5)}~~~~~\max_{\tilde N_{\mathrm{act}}}
    \quad~~\xi_1\left(-\tilde N_{\mathrm{act}}+\xi_2\sqrt{\tilde N_{\mathrm{act}}}+\xi_3\right)^2\label{obj_p3}\\
    &\qquad~~~\!~\quad~\text{s.t.} ~~~~~0< \tilde N_{\mathrm{act}}\leq\frac{W_0}{W_{\mathrm{act}}},
\end{align}
where
\begin{align}
    \xi_1=\frac{P_{\mathrm{B}}\beta^2W_{\mathrm{act}}^2/D_{\mathrm{BI}}^2d_{\mathrm{IU}}^2W_{\mathrm{pas}}^2}{A_{\mathrm{sum}}\sigma_{\mathrm{I}}^2\beta/d_{\mathrm{IU}}^2+\sigma_0^2},
    \qquad \xi_2 = \frac{\sqrt{A_{\mathrm{sum}}}W_{\mathrm{pas}}}{W_{\mathrm{act}}},
    \qquad \text{and } \xi_3 = \frac{W_0}{W_{\mathrm{act}}}.
\end{align}
\begin{theorem}\label{the_opt_N}
\textbf{\emph{(Optimal active/passive elements allocation)}} \emph{Under the LoS channel model and given $\tilde N_{\mathrm{act}}>0$, the optimal solution to problem (P4) is}
\begin{equation}\label{opt_na}
\begin{cases}
&\tilde N^*_{\mathrm{act}}=\frac{W_0}{W_{\mathrm{act}}},\tilde N^*_{\mathrm{pas}}=0, \qquad\qquad\qquad\qquad\qquad\qquad~~~\emph{if } W_0< \frac{W_{\mathrm{pas}}^{2} P_{\mathrm{I}} / W_{\mathrm{act}}}{4 P_{\mathrm{B}} \beta / D_{\mathrm{BI}}^{2}+4 \sigma_{\mathrm{I}}^{2}},\\
&\tilde N^*_{\mathrm{act}}=\frac{P_{\mathrm{I}}W_{\mathrm{pas}}^2/W_{\mathrm{act}}^2}{4P_{\mathrm{B}}\beta/D_{\mathrm{BI}}^2+4\sigma_{\mathrm{I}}^2},\tilde N^*_{\mathrm{pas}}=\frac{W_0}{W_{\mathrm{pas}}}-\frac{P_{\mathrm{I}}W_{\mathrm{pas}}/W_{\mathrm{act}}}{4P_{\mathrm{B}}\beta/D_{\mathrm{BI}}^2+4\sigma_{\mathrm{I}}^2},\qquad \emph{otherwise}.
\end{cases}
\end{equation}
\end{theorem}
\begin{proof}
See Appendix \ref{proof_lem3}.
\end{proof}
\begin{remark}
\emph{Theorem \ref{the_opt_N} shows that when the total deployment budget is small, i.e., $W_0< \frac{W_{\mathrm{pas}}^{2} P_{\mathrm{I}} / W_{\mathrm{act}}}{4 P_{\mathrm{B}} \beta / D_{\mathrm{BI}}^{2}+4 \sigma_{\mathrm{I}}^{2}}$, only active elements should be employed since they can provide a signal amplification gain and thus a higher rate than passive elements. In contrast, if the total deployment budget is sufficiently large, i.e., $W_0\geq \frac{W_{\mathrm{pas}}^{2} P_{\mathrm{I}} / W_{\mathrm{act}}}{4 P_{\mathrm{B}} \beta / D_{\mathrm{BI}}^{2}+4 \sigma_{\mathrm{I}}^{2}}$, the optimal number of active elements should balance the amplification gain of active elements and the beamforming gain of passive elements, which is independent of $W_0$ but determined by other parameters, such as $P_{\mathrm{I}}$, $W_{\mathrm{act}}$ and $W_{\mathrm{pas}}$.
This is because the performance bottleneck of active elements becomes the limited amplification power when $W_0$ is large, where the limited power can only support partial active elements to operate with the optimal amplification factor.}
\end{remark}

Based on Theorem \ref{the_opt_N} and substituting the optimal number of active elements, $\tilde N_{\mathrm{act}}=\frac{P_{\mathrm{I}} W_{\mathrm{pas}}^{2} / W_{\mathrm{act}}^{2}}{4 P_{\mathrm{B}} \beta / D_{\mathrm{BI}}^{2}+4 \sigma_{\mathrm{I}}^{2}}$, into \eqref{C_h}, the achievable capacity of the worst-case user is obtained as
\begin{equation}
C^*_{\text {L,hyb}} = \log _{2}\left(1+\frac{P_{\mathrm{B}} \beta^{2}\left(\frac{A_{\text {sum }} W_{\mathrm{pas}}}{4 W_{\mathrm{act}}}+\frac{W_0}{W_{\mathrm{pas}}}\right)^{2} / D_{\mathrm{BI}}^{2} d_{\mathrm{IU}}^{2}}{A_{\text {sum }} \sigma_{\mathrm{I}}^{2} \beta / d_{\mathrm{IU}}^{2}+\sigma_{0}^{2}}\right).\label{C_h_2}
\end{equation}

\begin{corollary}
\emph{For the case of LoS channels, given the optimal IRS phase shifts in \eqref{opt_phase_1} and \eqref{opt_phase_2} and active/passive elements allocation in \eqref{opt_na}, the optimal amplification factor of each active element is given by}
\begin{equation}
    \alpha_{n}=\frac{2 W_{\mathrm{act}}}{W_{\mathrm{pas}}} ,\forall  n \in\mathcal{N}_{\mathrm{act}},
\end{equation}
\emph{for achieving the capacity in \eqref{C_h_2}.}
\end{corollary}

Moreover, when $W_0<\frac{W_{\mathrm{pas}}^2P_{\mathrm{I}}/W_{\mathrm{act}}}{4P_{\mathrm{B}}\beta/D_{\mathrm{BI}}^2+4\sigma_{\mathrm{I}}^2}$ and hence $N^*_{\mathrm{act}}=W_0/W_{\mathrm{act}}$, the achievable capacity of the worst-case user aided by the active IRS under the LoS channel model is given by
\begin{equation}
    C^*_{\text {L,act}} \triangleq \log _{2}\left(1+\frac{W_0A_{\mathrm{sum}} P_{\mathrm{B}} \beta^{2} /W_{\mathrm{act}} D_{\mathrm{BI}}^{2} d_{\mathrm{IU}}^{2}}{A_{\mathrm{sum}} \sigma_{\mathrm{I}}^{2} \beta / d_{\mathrm{IU}}^{2}+\sigma_{0}^{2}}\right).\label{C_a}
\end{equation}

Next, we compare the achievable capacity of the hybrid IRS with that of the fully-active and fully-passive IRSs with respect to (w.r.t.) different values of the deployment budget under the LoS channels. To this end, we define the following active/passive elements allocation thresholds for the deployment budget.
\begin{align}
    W_{\mathrm{A-H}}&\triangleq\frac{W_{\mathrm{pas}}^2P_{\mathrm{I}}/W_{\mathrm{act}}}{4P_{\mathrm{B}}\beta/D_{\mathrm{BI}}^2+4\sigma_{\mathrm{I}}^2},\label{B_H-A}\\
    W_{\mathrm{H-P}}&\triangleq\frac{W_{\mathrm{pas}}^2\sigma_0^2d_{\mathrm{IU}}^2}{4W_{\mathrm{act}}\sigma_{\mathrm{I}}^2\beta}+\frac{W_{\mathrm{pas}}^2\sigma_0d_{\mathrm{IU}}}{4W_{\mathrm{act}}\sigma_{\mathrm{I}}}\sqrt{\frac{\sigma_0^2d_{\mathrm{IU}}^2}{\sigma_{\mathrm{I}}^2\beta^2}+\frac{P_{\mathrm{I}}}{P_{\mathrm{B}}\beta^2/D_{\mathrm{BI}}^2+\sigma_{\mathrm{I}}^2\beta}}.\label{condition_H-P}
\end{align}
Then, we have the following result.
\begin{theorem}\label{lem_thres}
\textbf{\emph{(Capacity comparison among different IRS architectures)}}
\emph{Under the LoS channel model, the capacity comparison among different IRS architectures is given as follows.}

\emph{1) When $0\leq W_0< W_{\mathrm{A-H}}$, $C^*_{\text {L,act}}=C^*_{\text {L,hyb}}>C^*_{\text {L,pas}}$;}

\emph{2) When $W_{\mathrm{A-H}}\leq W_0\leq W_{\mathrm{H-P}}$, $C^*_{\text {L,hyb}}>C^*_{\text {L,act}}$ and $C^*_{\text {L,hyb}}>C^*_{\text {L,pas}}$;}

\emph{3) When $W_0>W_{\mathrm{H-P}}$, $C^*_{\text {L,pas}}=C^*_{\text {L,hyb}}>C^*_{\text {L,act}}$.}
\end{theorem}
\begin{proof}
See Appendix \ref{proof_lem4}.
\end{proof}
\begin{remark}
\textbf{\emph{(IRS architecture selection)}}
\emph{Theorem \ref{lem_thres} shows that the optimal architecture selection for the IRS (i.e., passive, active, or hybrid) is determined by the total deployment budget. Specifically, when $W_0$ is small, the hybrid IRS reduces to the fully-active IRS, which is the optimal architecture to achieve the maximum capacity. This is because when $N_{\mathrm{pas}}$ is small, the active IRS can provide a high power amplification gain while passive elements only have limited beamforming gain. In contrast, when $W_0$ is sufficiently large, the hybrid IRS reduces to the fully-passive IRS. This is expected since in this case, the amplification factor for the active elements is limited, which may not be able to mitigate the amplification noise. However, when $W_{\mathrm{A-H}}\leq W_0\leq W_{\mathrm{H-P}}$, there in general exists a trade-off between increasing the power amplification gain (i.e., assigning more active elements) and beamforming gain (i.e., assigning more passive elements); hence, it is necessary to design the optimal active/passive elements allocation for the hybrid IRS to maximize the capacity.}
\end{remark}
\begin{remark}
\textbf{\emph{(Deployment budget thresholds)}}
\emph{It is observed from \eqref{B_H-A} and \eqref{condition_H-P} that $W_{\mathrm{A-H}}$ and $W_{\mathrm{H-P}}$ both increase with the amplification power, $P_{\mathrm{I}}$, and the passive-element cost, $W_{\mathrm{pas}}$, while they decrease with the active-element cost, $W_{\mathrm{act}}$. 
This can be explained as follows. 
First, a higher amplification power, $P_{\mathrm{I}}$, allows more active elements to operate with sufficiently large amplification factors. 
Second, higher passive-element cost, $W_{\mathrm{pas}}$, and/or lower active-element cost, $W_{\mathrm{act}}$, make it more desirable to deploy active elements.}
\end{remark}

\subsubsection{Rayleigh fading channel model}
For the case of Rayleigh fading channels with $K_1=K_2=0$, we obtain from \eqref{gamma_approx} that $\gamma_1=0$ and $\gamma_2=1$. As such, the ergodic capacity under the Rayleigh fading channels is given by
\begin{equation}
    \tilde C=C_{\mathrm{NL}}\triangleq\log _{2}\left(1+\frac{\left(A_{\mathrm{sum}}+{N}_{\mathrm{pas}}\right) \frac{P_{\mathrm{B}} \beta^{2}}{D_{\mathrm{BI}}^{2} d_{\mathrm{IU}}^{2}}}{A_{\mathrm{sum}} \sigma_{\mathrm{I}}^{2} \beta / d_{\mathrm{IU}}^{2}+\sigma_{0}^{2}}\right).
\end{equation}

\begin{lemma}\label{lem_opt_ea_NLoS}
\textbf{\emph{(Optimal active/passive elements allocation for Rayleigh fading channels)}}
\emph{Under the Rayleigh fading channel model and favorable amplification power condition in \eqref{cons_pb}, the optimal number of active and passive elements are given by}
\begin{equation}\label{opt_na_NLoS}
\begin{cases}
& N^*_{\mathrm{act}}=1, N^*_{\mathrm{pas}}=\lfloor\frac{W_0-W_{\mathrm{act}}}{W_{\mathrm{pas}}}\rfloor, \qquad~~~\emph{if } W_0> \frac{A_{\mathrm{sum}}W_{\mathrm{pas}}\sigma_0^2-W_{\mathrm{act}}\sigma_0^2}{A_{\mathrm{sum}}\sigma_{\mathrm{I}}^2\beta/d_{\mathrm{IU}}^2},\\
& N^*_{\mathrm{act}}=0, N^*_{\mathrm{pas}}=\lfloor\frac{W_0}{W_{\mathrm{pas}}}\rfloor,\qquad~~~~~~~\emph{otherwise}.
\end{cases}
\end{equation}
\end{lemma}
\begin{proof}
First, when $N_{\mathrm{act}}>0$, it can be shown that
\begin{equation}
    C_{\mathrm{NL}}\leq C^*_{\mathrm{NL,hyb}}\triangleq\log _{2}\left(1+\frac{\left(A_{\mathrm{sum}}+{N}_{\mathrm{pas,max}}\right) \frac{P_{\mathrm{B}} \beta^{2}}{D_{\mathrm{BI}}^{2} d_{\mathrm{IU}}^{2}}}{A_{\mathrm{sum}} \sigma_{\mathrm{I}}^{2} \beta / d_{\mathrm{IU}}^{2}+\sigma_{0}^{2}}\right),\label{C_NLoS_hyb}
\end{equation}
where ${N}_{\mathrm{pas,max}}=\lfloor\frac{W_0-W_{\mathrm{act}}}{W_{\mathrm{pas}}}\rfloor$ and $N_{\mathrm{act}}=1$.
Second, when $N_{\mathrm{act}}=0$, we obtain that $A_{\mathrm{sum}}=0$ and the ergodic capacity achieved by passive IRS is given by
\begin{equation}
    C_{\mathrm{NL}}=C^*_{\mathrm{NL,pas}}\triangleq\log _{2}\left(1+\frac{W_0P_{\mathrm{B}} \beta^{2}}{W_{\mathrm{pas}} \sigma_{0}^{2} D_{\mathrm{BI}}^{2} d_{\mathrm{IU}}^{2}}\right). \label{C_NLoS_pas}
\end{equation}
Then, by comparing $C^*_{\mathrm{NL,hyb}}$ in \eqref{C_NLoS_hyb} and $C^*_{\mathrm{NL,pas}}$ in \eqref{C_NLoS_pas}, it can be shown that when $W_0> \frac{A_{\mathrm{sum}}W_{\mathrm{pas}}\sigma_0^2-W_{\mathrm{act}}\sigma_0^2}{A_{\mathrm{sum}}\sigma_{\mathrm{I}}^2\beta/d_{\mathrm{IU}}^2}$, we have $C^*_{\mathrm{NL,hyb}}>C^*_{\mathrm{NL,pas}}$. Based on the above, the optimal active/passive elements allocation under the Rayleigh fading channel model is given in \eqref{opt_na_NLoS}, thus completing the proof.
\end{proof}
\begin{remark}
\textbf{\emph{(Active/passive elements allocation under the Rayleigh fading channel model)}}
\emph{Lemma \ref{lem_opt_ea_NLoS} shows that under the Rayleigh fading channel model (or rich scattering environment), most of the deployment budget should be assigned to passive elements, while at most one active element needs to be deployed.
Specifically, when the amplification power is sufficiently small and/or the deployment budget is sufficiently large, all the deployment budget should be assigned to passive elements. This is because the power amplification gain provided by active elements cannot mitigate their amplification noise and/or the beamforming gain of passive elements (which prevails active-element power amplification).
On the other hand, when the amplification power is large and/or  the deployment budget is small, one active element is enough to achieve the amplification power gain because it has no active-element beamforming gain under the Rayleigh fading channels, thus making it desirable to assign most of the deployment budget to passive elements.
}
\end{remark}

\section{Simulation Results}\label{sec_sim}
Simulation results are presented in this section to evaluate the effectiveness of the proposed hybrid IRS architecture and active/passive elements allocation design. 
If not specified otherwise, the system setups are as follows.
We consider a two-dimensional (2D) Cartesian coordinate system, where the reference points of the BS, the hybrid IRS, and the worst-case user are located at $(0,0)$ m, $(60,0)$ m, and $(60,20)$ m, respectively. 

The deployment costs of each active/passive element are set as $W_{\mathrm{act}} = 5$ and $W_{\mathrm{pas}} = 1$, respectively, by taking into account the fact that the active element in general incurs higher static operation power and hardware cost.
For each active element, the feasible region of the amplification factor is in the range of $[\alpha_{\min},\alpha_{\max}]=[0,14]$ dB \cite{8403249}.
The system operates at a carrier frequency of 6 GHz with the wavelength of $\lambda=0.05 \mathrm{~m}$. 
We consider the same Rician fading factor for the BS-IRS and IRS-user channels, i.e., $K_1=K_2=K$.
Other parameters are set as $d_{\mathrm{I}}=\frac{\lambda}{4}$, $\beta=-30$ dB, $\sigma_{\mathrm{I}}=\sigma_0=-80$ dBm, $P_{\mathrm{B}}=15$ dBm, and $P_{\mathrm{I}} = 5$ dBm \cite{9377648}.
\subsection{Accuracy of Ergodic Capacity Approximation}
\begin{figure}[t]
\centerline{\includegraphics[width=2.7in]{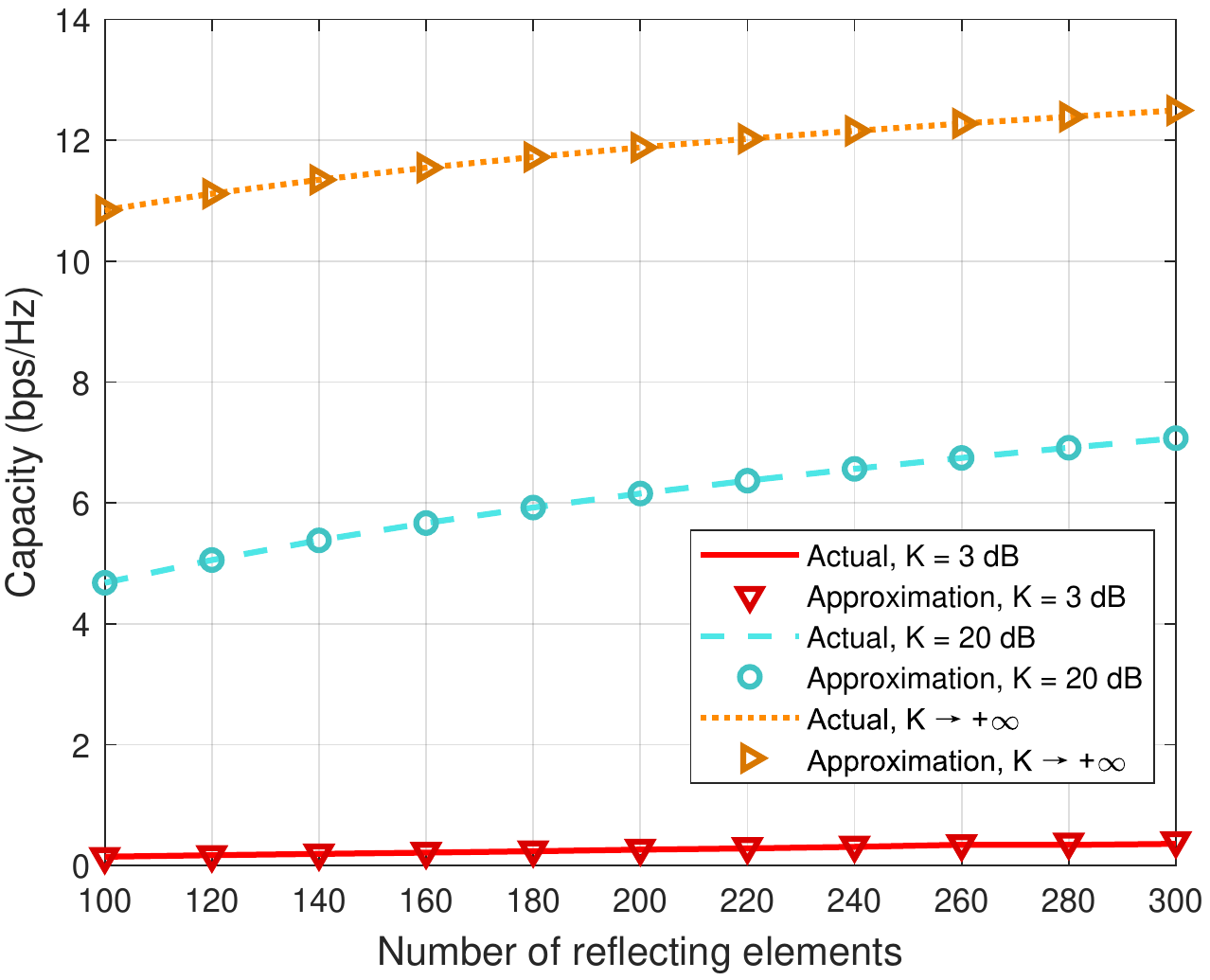}}
\caption{Accuracy of the ergodic capacity approximation given in \eqref{sig_approx}.}\label{eval_approx}
\end{figure}

First, we evaluate in Fig. \ref{eval_approx} the accuracy of the ergodic capacity approximation presented in Lemma \ref{lem_C_approx} (see \eqref{sig_approx}). 
By setting equal active/passive elements allocation, i.e., $N_{\mathrm{act}}=N_{\mathrm{pas}}$, the actual ergodic capacity in \eqref{ergo_capa} is obtained based on 1000 independent channel realizations under different Rician factors of $K\in\{0$ dB, $10$ dB, $+\infty\}$ ($K\to +\infty$ corresponds to the LoS channels).
It is observed that the approximated ergodic capacity in \eqref{sig_approx} is close to the exact capacity in \eqref{ergo_capa} under different Rician factors. This is because when the number of reflecting elements is large, the variance of random channels is averaged due to the law of large numbers.
\subsection{Effect of Active/passive Elements Allocation on Ergodic Capacity}
\begin{figure}[ht]
\centerline{\includegraphics[width=2.7in]{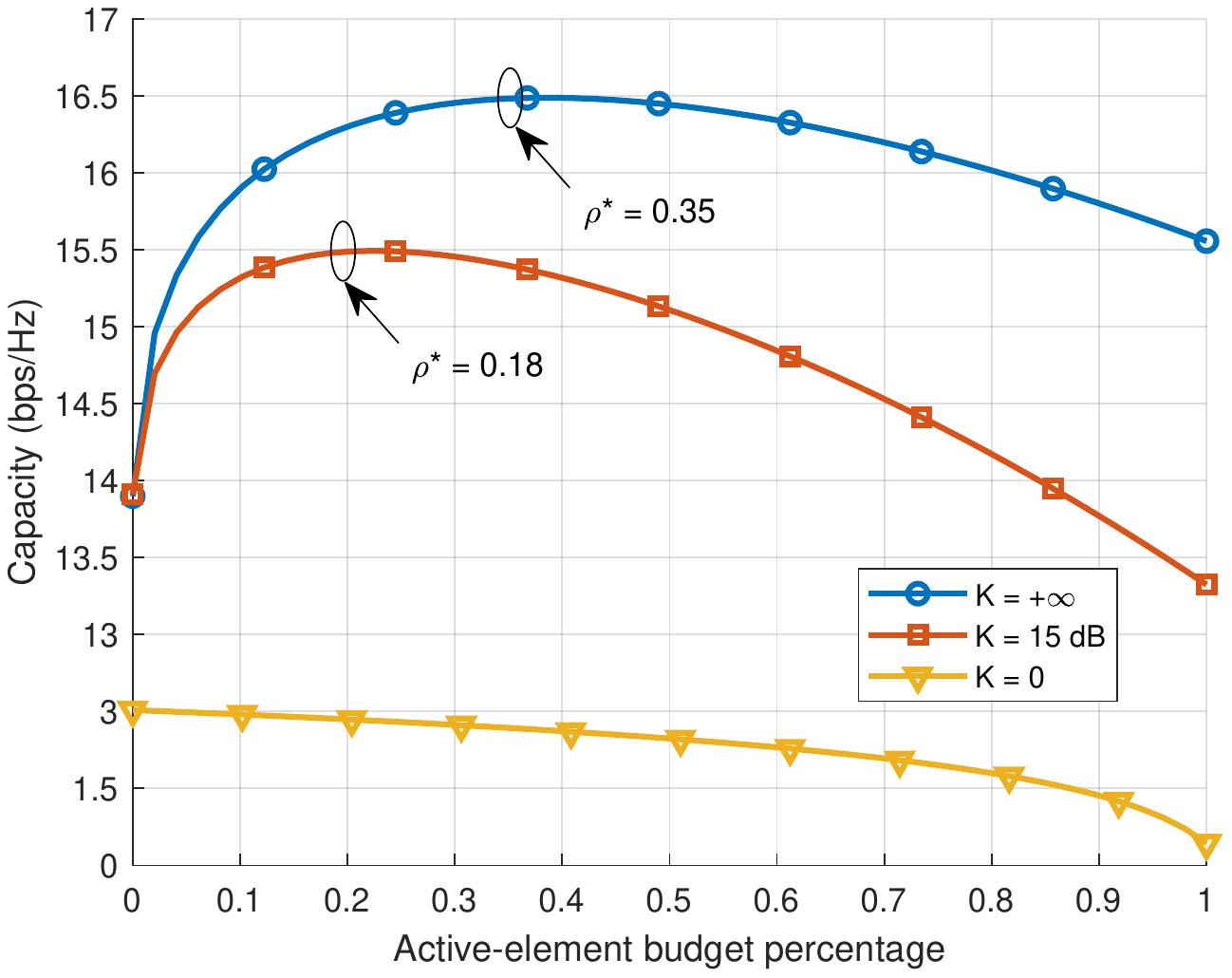}}
\caption{Capacity performance versus budget percentage of active elements.}\label{rate_ratio}
\end{figure}

Next, we investigate the effect of the active/passive elements allocation at the hybrid IRS on the capacity performance.
We set the total deployment budget $W_0=3000$, the amplification power $P_{\mathrm{I}}=15$ dBm, and denote $\rho$ as the percentage of the budget assigned to active elements ($0\leq \rho\leq 1$).
In Fig. \ref{rate_ratio}, we plot the ergodic capacity versus $\rho$ with different Rician factors $K\in\{0$, $15$ dB, $+\infty\}$.
It is observed that for the cases of $K\to\infty$ and $K=15$ dB (with LoS channel components), the ergodic capacity first increases with $\rho$, then decreases after it exceeds a threshold (i.e., $\rho^*=0.35$ for $K\to\infty$, and $\rho^*=0.18$ for $K=15$ dB). This shows that for the general Rician fading channels, the active/passive elements allocation has a significant effect on the capacity maximization.
This is expected because if the budget assigned to active elements is too small, the active-element power amplification gain is not fully exploited due to the small number of active elements. However, if the budget assigned to active elements is too large, the performance bottleneck of active elements becomes the limited amplification power that cannot support all active elements to operate in the amplification mode, thus making it desirable to assign partial budget to passive elements for achieving higher beamforming gain.
Moreover, the optimal budget assigned to active elements increases with the Rician factor because the active elements can achieve both the power amplification gain and beamforming gain over the LoS paths but have no beamforming gain over NLoS paths.
Last, when $K=0$ (corresponding to the Rayleigh fading channels), the total deployment budget should be assigned to passive elements to maximize the ergodic capacity since the beamforming gain of passive elements is larger than the amplification gain of active elements.
\subsection{Effect of Total Deployment Budget}
\begin{figure}[t] \centering    
{\subfigure[{LoS channels with $K\to +\infty$.}] {
\label{fig_3_LoS}
\includegraphics[width=2.7in]{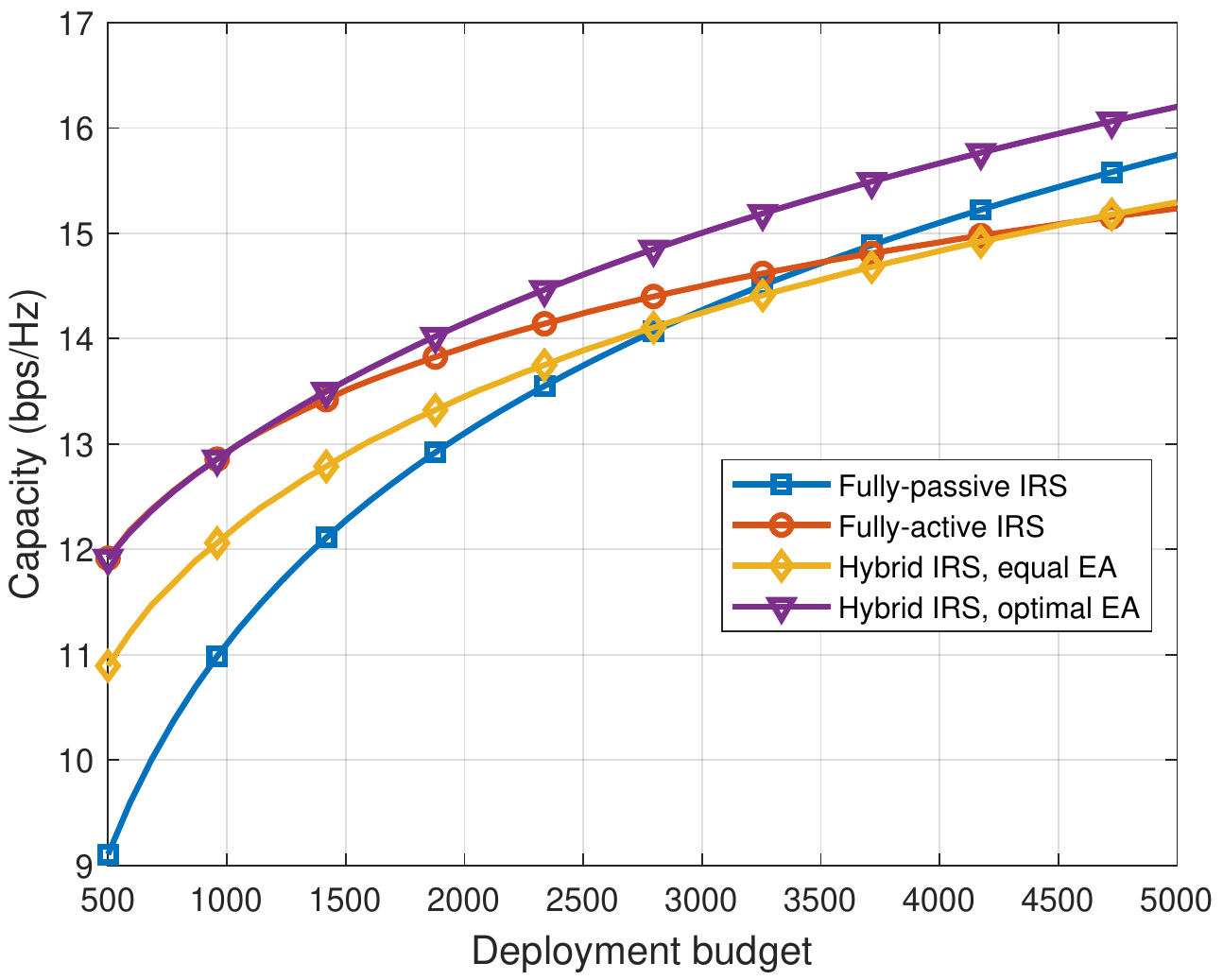}  
}}     
{\subfigure[{Rician fading channels with $K=10$ dB.}] {\label{fig_3_Rician}
\includegraphics[width=2.7in]{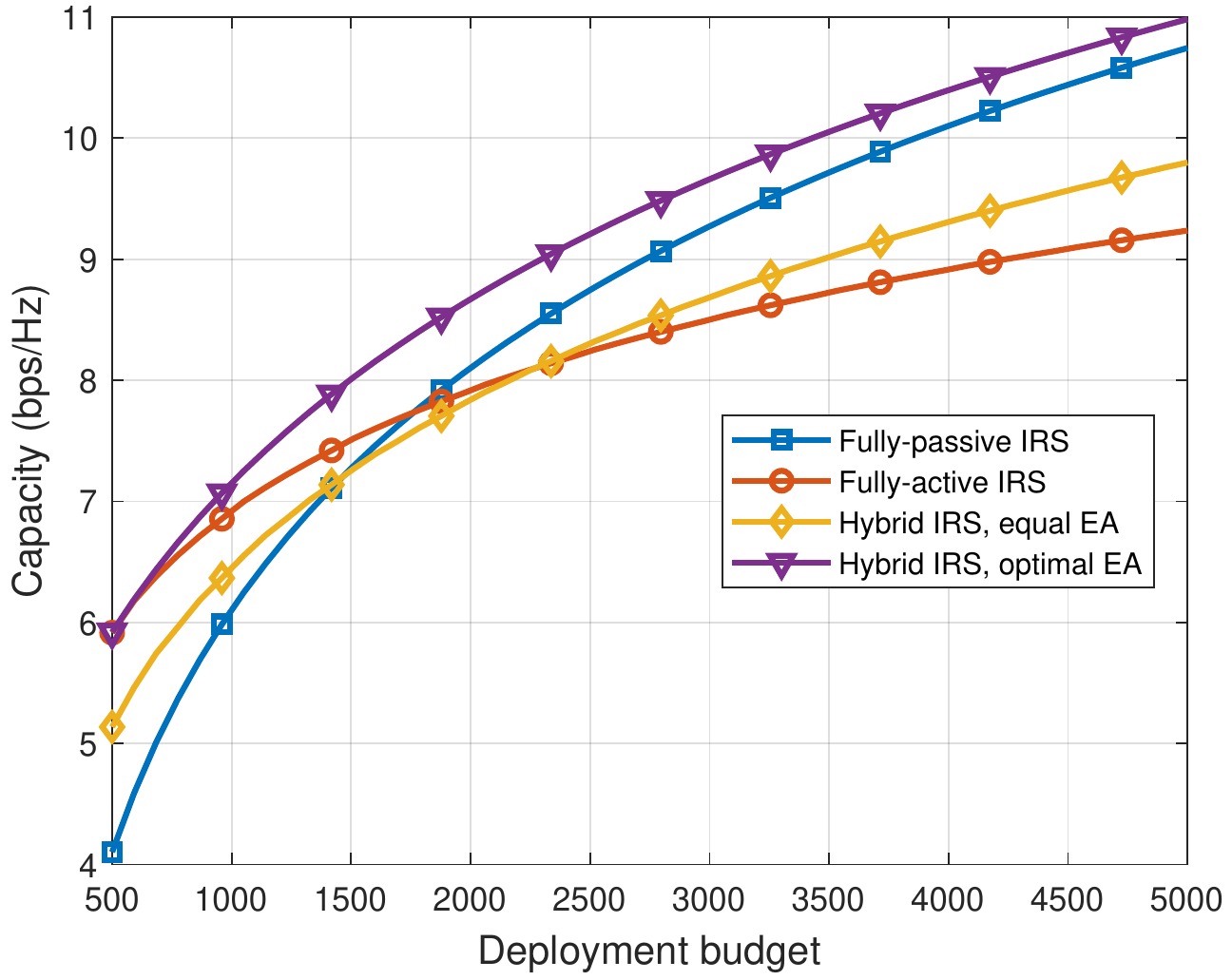}  
}}
{\caption{The ergodic capacity achieved by hybrid IRS, fully-active IRS, and fully-passive IRS versus deployment budget.}}
\end{figure}

Next, we compare the proposed optimal active/passive elements allocation (EA) design against three benchmarks: 1) Fully-active IRS with all the deployment budget assigned to active elements; 2) Fully-passive IRS with all the deployment budget assigned to active elements; 3) Hybrid IRS under equal EA for which the deployment budget is equally assigned to active and passive elements. We apply the optimal IRS beamforming design in \eqref{opt_phase_1}, \eqref{opt_phase_2}, and \eqref{opt_a_n} for the proposed hybrid IRS architecture and the benchmarks.

Figs. \ref{fig_3_LoS} and \ref{fig_3_Rician} show the ergodic capacity of the worst-case user versus the total deployment budget under the LoS and Rician fading channel models, respectively.
First, it is observed that the ergodic capacity achieved by the hybrid IRS architecture with optimal active/passive elements allocation is always larger than or equal to that achieved by the fully-active or fully-passive IRSs. This is expected because the hybrid IRS provides an extra degree of freedom for elements allocation to balance the trade-off between the power amplification and beamforming gains.
Second, the hybrid IRS with the optimal active/passive elements allocation reduces to the fully-active IRS when the deployment budget is sufficiently small (i.e., $W_0<1250$ for the LoS channels, and $W_0<600$ for the Rician fading channels with $K=10$ dB). 
This is because the active elements has a higher power amplification gain as compared to the passive-element beamforming gain when the budget is small.
Third, one can observe that the hybrid IRS with the optimal active/passive elements allocation outperforms that with equal elements allocation, which shows the effectiveness of elements allocation optimization for the hybrid IRS.
Last, it is also observed that given the same deployment budget, the achievable capacity under the LoS channels is higher than that under Rician fading channels since the IRS phase shifts are designed based on the LoS channel components or statistical CSI only.

\begin{figure}[t] \centering    
{\subfigure[{LoS channels with $K\to +\infty$.}] {
\label{fig_4_LoS}
\includegraphics[width=2.7in]{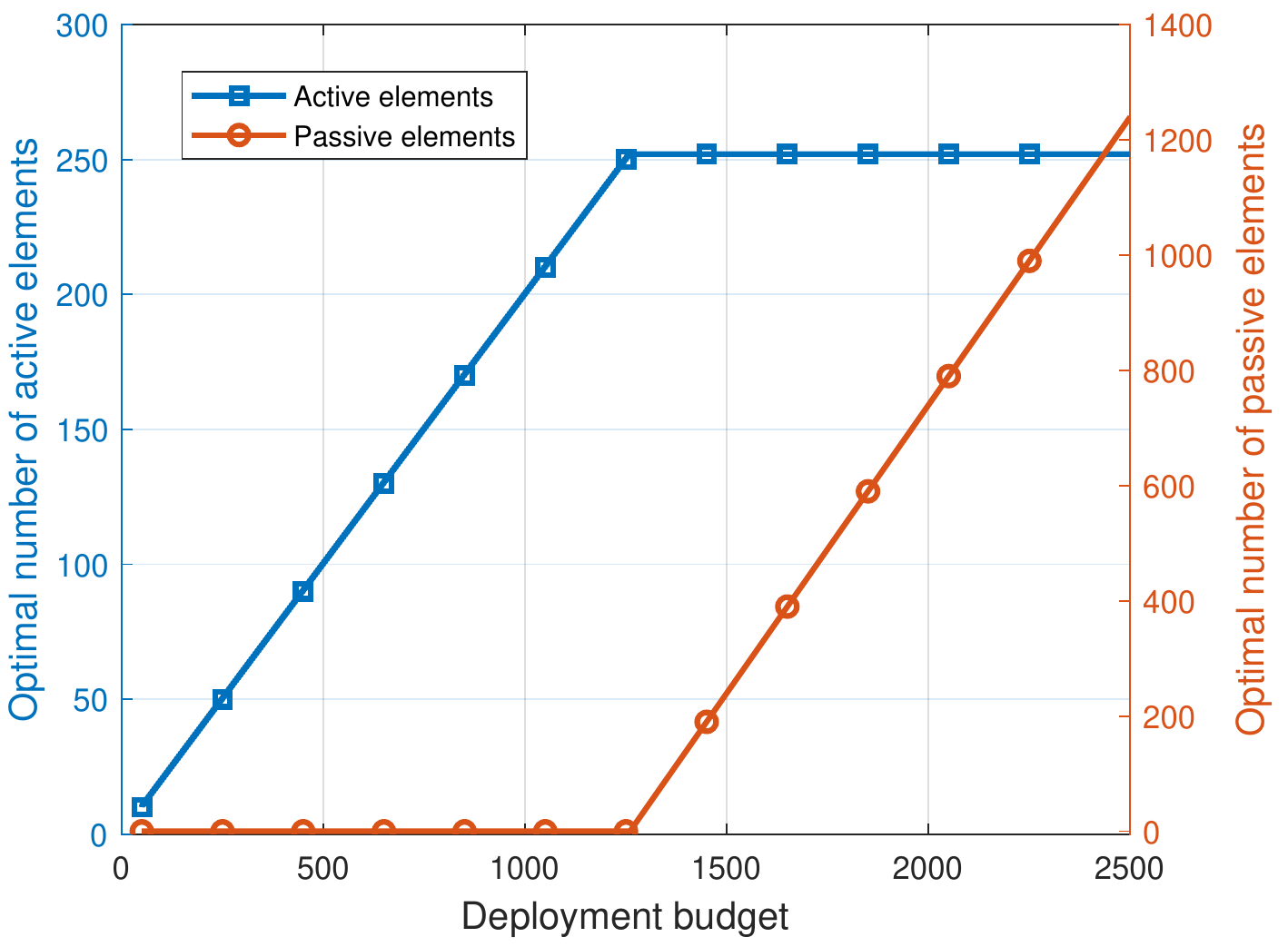}  
}}     
{\subfigure[{Rician fading channels with $K=10$ dB.}] {\label{fig_4_Rician}
\includegraphics[width=2.7in]{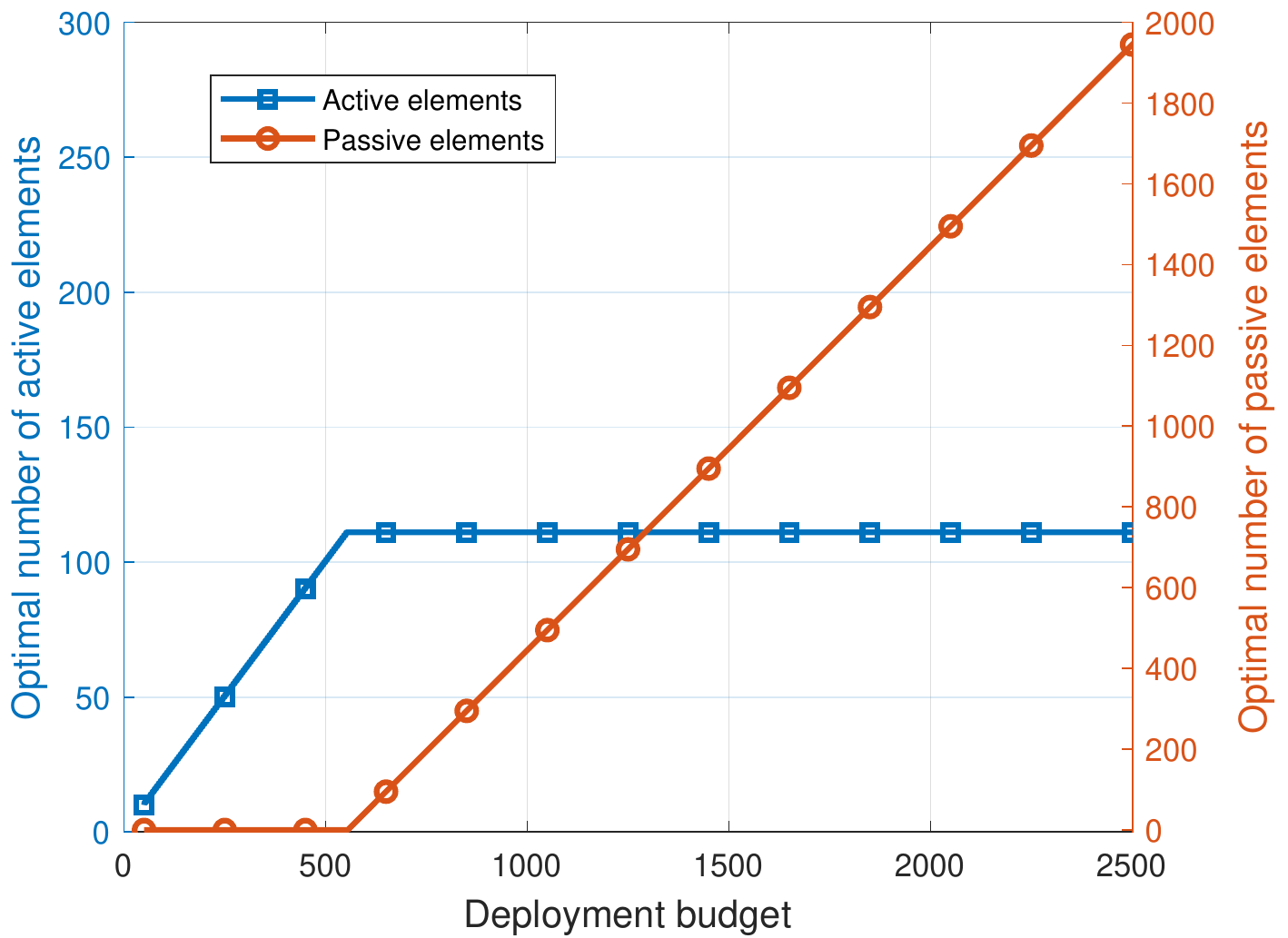}     
}}
{\caption{The optimal numbers of active and passive elements versus total deployment budget.}}
\end{figure}

Figs. \ref{fig_4_LoS} and \ref{fig_4_Rician} present the optimal numbers of active and passive elements for maximizing the ergodic capacity of the worst-case user under different total deployment budgets with $W_{\mathrm{act}}=5$ and $W_{\mathrm{pas}}=1$.
It is observed that as $K$ increases, the deployment budget is first assigned to active elements only and then assigned to more passive elements after it exceeds a threshold.
In addition, we observe that given the same deployment budget, 
the optimal number of active elements for the LoS channels, i.e., $N^*_{\mathrm{act}}=250$ is higher than that for the Rician fading channels with $K=10$ dB, i.e., $N^*_{\mathrm{act}}=120$.
\subsection{Effect of Rician Factor}
In Figs. \ref{fig5_capa_vs_K} and \ref{fig5_ea_vs_K}, we investigate the effect of the Rician factor of the BS-IRS and IRS-user channels on the achievable capacity and active/passive elements allocation. The IRS beamforming design and the active-element amplification factors are given by \eqref{opt_phase_1}, \eqref{opt_phase_2}, and \eqref{opt_a_n}, respectively.

\begin{figure}[t] \centering    
{\subfigure[{Ergodic capacity versus Rician factor.}] {
\label{fig5_capa_vs_K}
\includegraphics[width=2.63in]{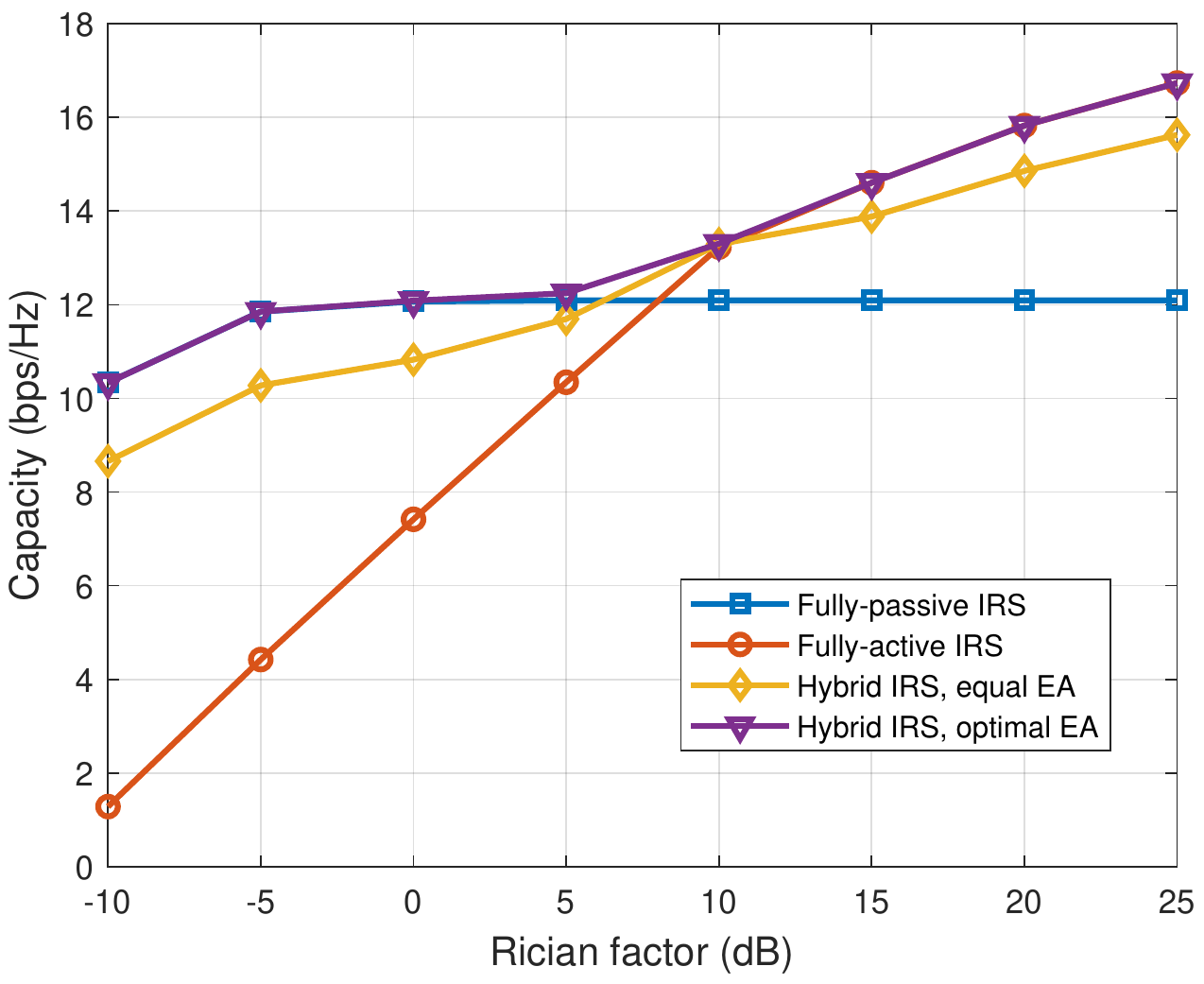}  
}}     
{\subfigure[Optimal number of active and passive elements versus Rician factor.] {\label{fig5_ea_vs_K}
\includegraphics[width=2.93in]{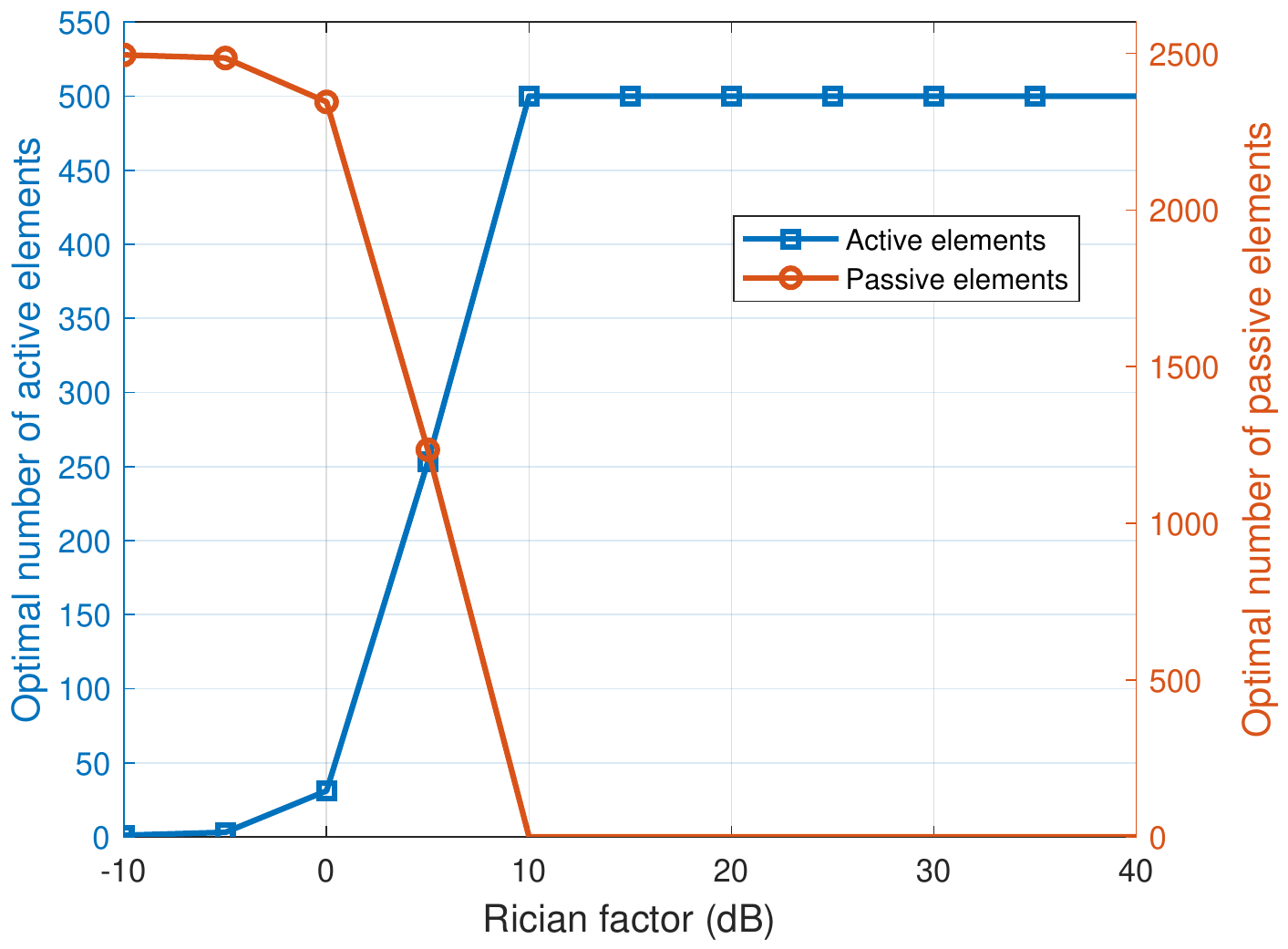}     
}}
{\caption{Effect of the Rician factor on ergodic capacity and active/passive elements allocation.}}
\end{figure}

Specifically, Fig. \ref{fig5_capa_vs_K} presents the ergodic capacity of the worst-case user under different Rician factors. It is observed that the proposed hybrid IRS architecture with the optimal active/passive elements allocation outperforms other benchmarks under different channel conditions, while the active IRS achieves a higher capacity than the passive IRS only when the channels are LoS-dominant (i.e., Rician factor $K$ is large). 
Fig. \ref{fig5_ea_vs_K} plots the optimal numbers of active and passive elements of the hybrid IRS under different Rician factors. 
It is observed that with an increasing Rician factor, the optimal number of active elements first increases and then keeps unchanged after it exceeds a threshold, where the NLoS channel components are negligible as compared to the LoS channel components. 
The reason is that for the NLoS paths, the active elements have no beamforming gain but can achieve power amplification gain due to the high amplification power, while for the LoS paths, active elements can achieve both the power amplification and beamforming gains with multiple elements. 
Therefore, as $K$ increases, more deployment budget should be assigned to the active elements so as to achieve a higher beamforming gain of the active elements.

\subsection{Effect of Amplification Power}
Moreover, we evaluate the effect of the amplification power of active elements on the ergodic capacity and active/passive elements allocation of the proposed hybrid IRS architecture under the general Rician fading channel model with $K=15$ dB.

\begin{figure}[t] \centering    
{\subfigure[{Ergodic capacity versus amplification power.}] {
\label{fig6_capa_vs_P_I}
\includegraphics[width=2.63in]{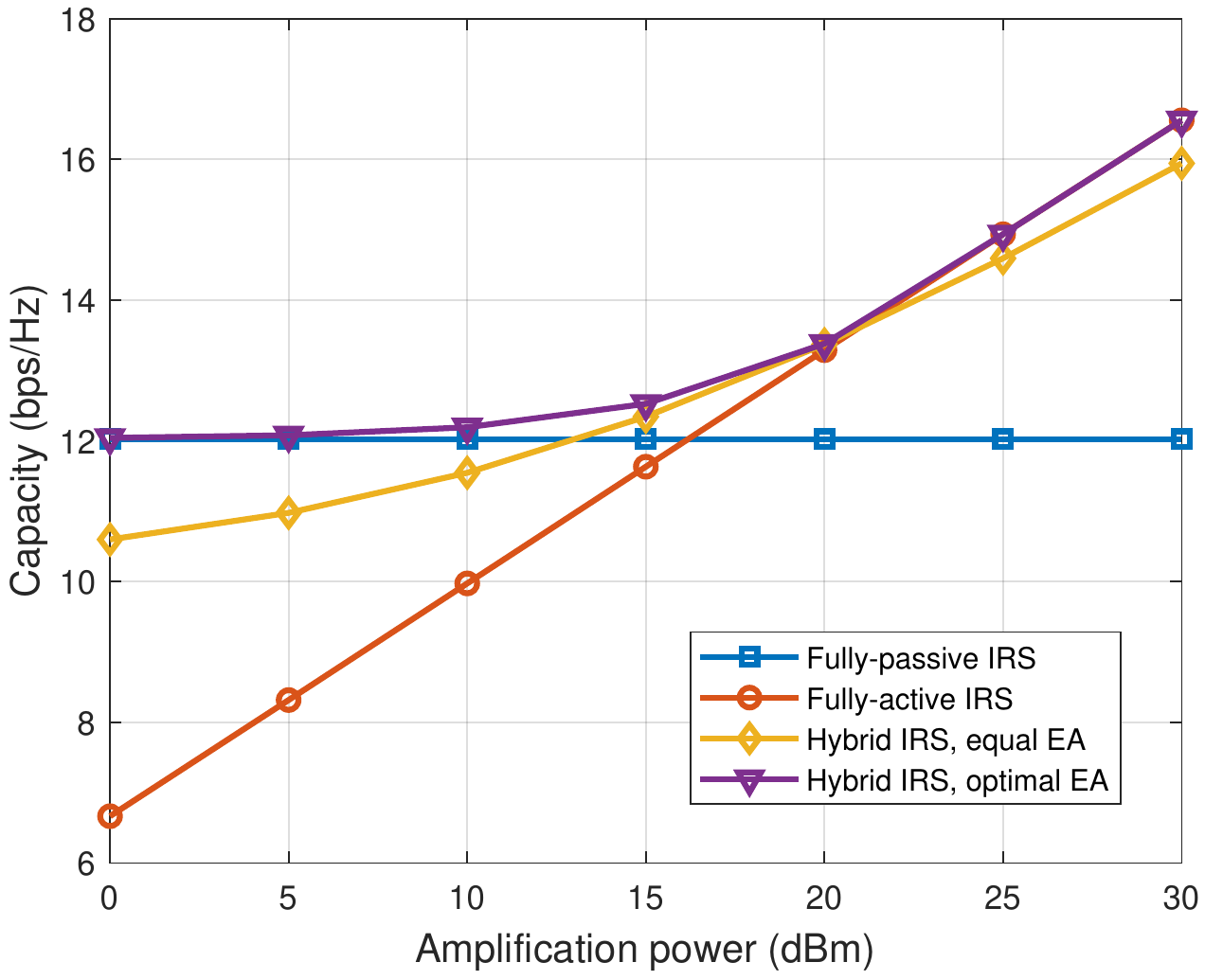}  
}}     
{\subfigure[Optimal numbers of active and passive elements versus amplification power.] {\label{fig6_ea_vs_P_I.pdf}
\includegraphics[width=2.93in]{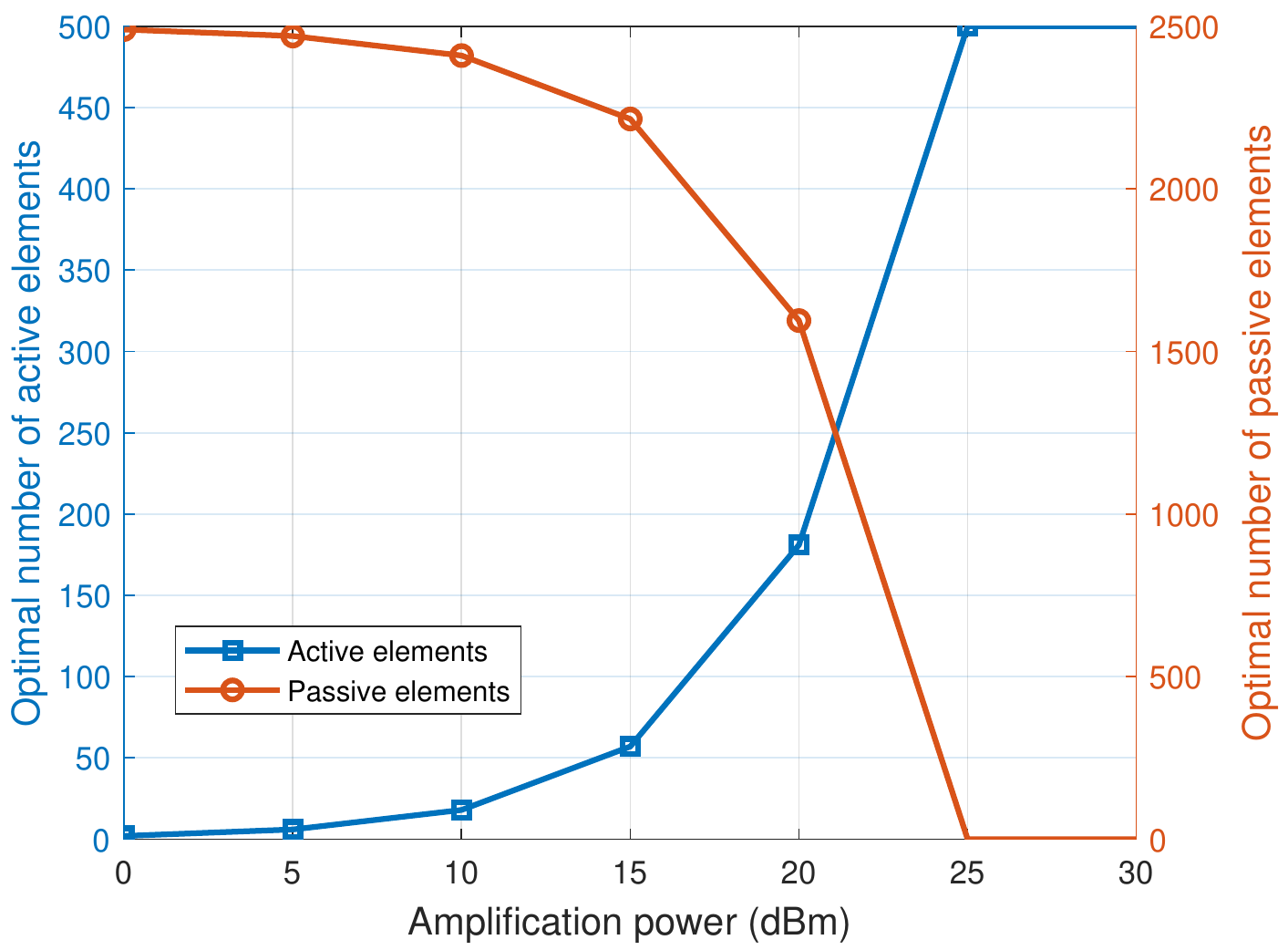}     
}}
{\caption{Effect of the amplification power on ergodic capacity and elements allocation.}}
\end{figure}

In Fig. \ref{fig6_capa_vs_P_I}, we plot the ergodic capacity versus the amplification power $P_{\mathrm{I}}$ for different schemes.
First, it is observed that the ergodic capacity of the worst-case user aided by the fully-active or hybrid IRSs monotonically increases with the amplification power. Second, we observe that the hybrid IRS with the optimal active/passive elements allocation achieves a much higher capacity than the fully-active IRS when $P_{\mathrm{I}}$ is small and significantly outperforms the fully-passive IRS when $P_{\mathrm{I}}$ is large. This shows the flexibility and advantages of the hybrid IRS under different amplification powers.

In Fig. \ref{fig6_ea_vs_P_I.pdf}, we plot the optimal numbers of active and passive elements for the hybrid IRS architecture versus the amplification power. 
One can observe that the optimal number of active elements increases with the amplification power, which is expected because a higher amplification power can support more active elements to operate in the amplification mode with optimal amplification factors. 
\subsection{Effect of Active/Passive-element Deployment Cost Ratio}
\begin{figure}[t] \centering    
{\subfigure[{Ergodic capacity versus active/passive-element cost ratio.}] {
\label{fig7_capa_vs_c_ratio}
\includegraphics[width=2.63in]{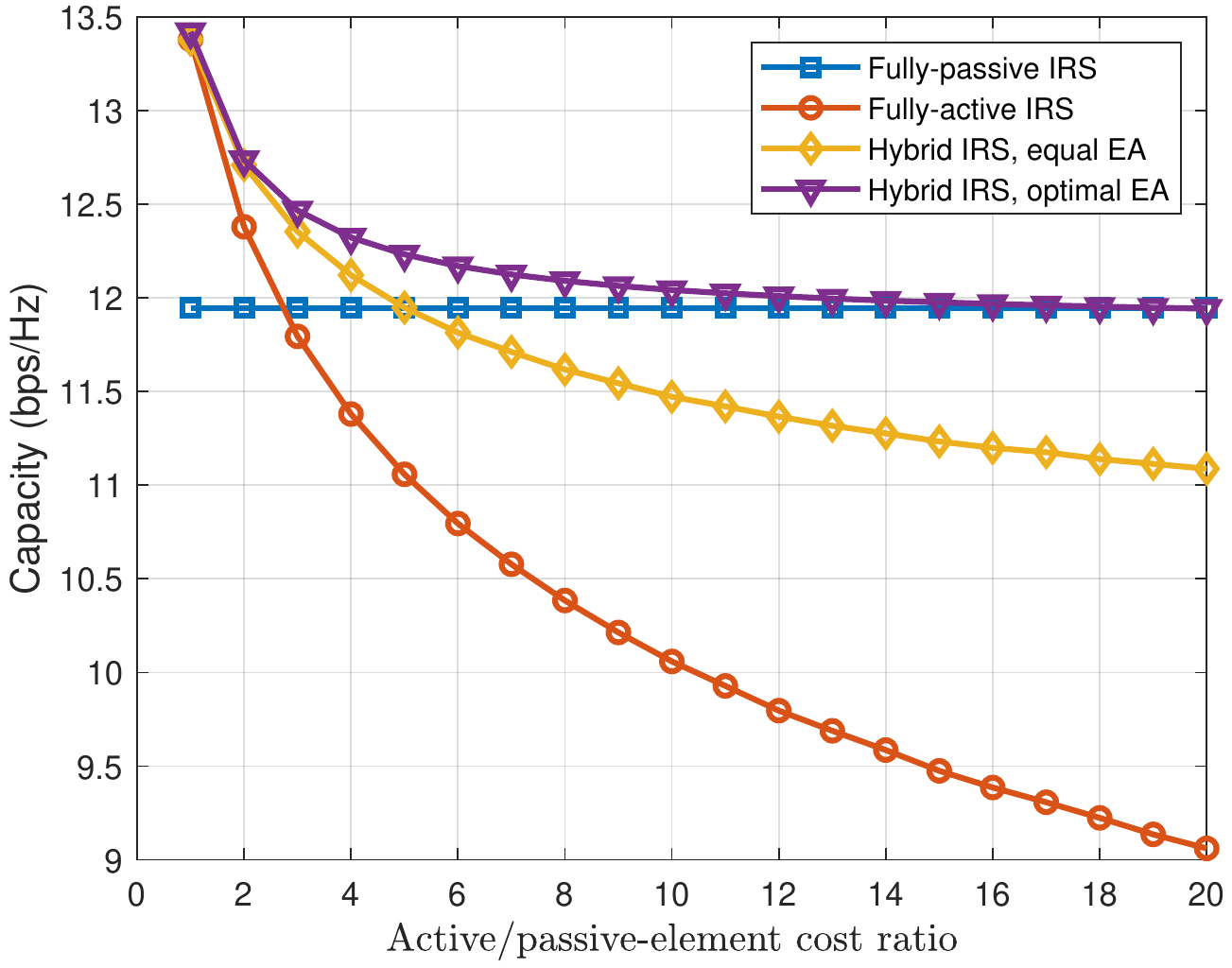}  
}}     
{\subfigure[Optimal numbers of active and passive elements versus deployment cost ratio.] {\label{fig7_ea_vs_c_ratio}
\includegraphics[width=2.93in]{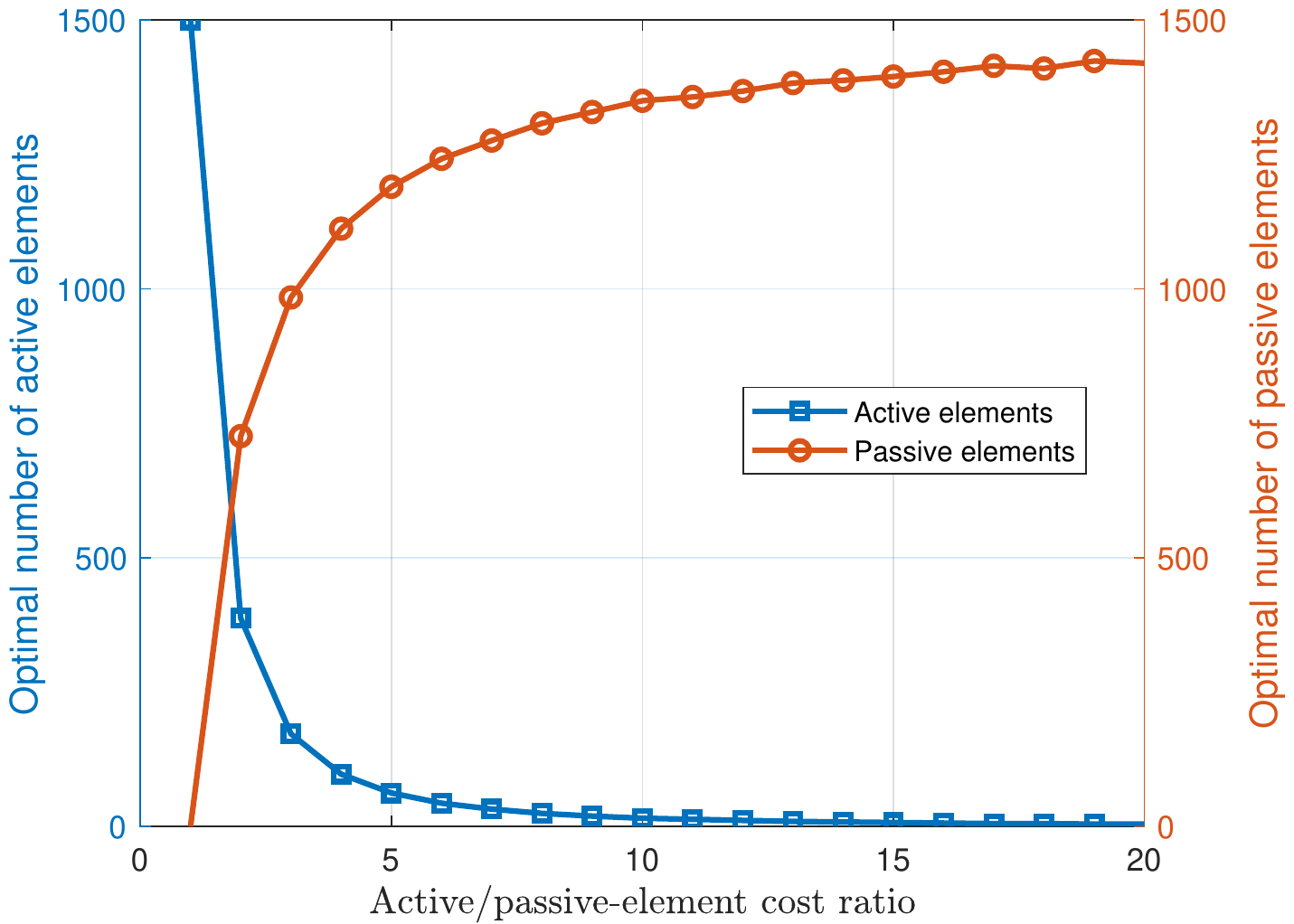}     
}}
{\caption{Effect of the deployment cost ratio on ergodic capacity and elements allocation.}}
\end{figure}
In Fig. \ref{fig7_capa_vs_c_ratio}, we compare the ergodic capacity of the worst-case user versus different ratios of active-over-passive deployment cost, i.e., $W_{\mathrm{act}}/W_{\mathrm{pas}}$, where we set $W_{\mathrm{pas}}=1$, $K=15$ dB, $W_0=1500$, and $P_{\mathrm{B}}=10$ dBm.
First, the hybrid IRS with the optimal active/passive elements allocation has the highest capacity as compared to other benchmarks. Besides, one can observe that as $W_{\mathrm{act}}$ increases, the ergodic capacity decreases when $W_{\mathrm{act}}$ is small, then remains unchanged when $W_{\mathrm{act}}$ is sufficiently large, which can be explained as follows.
When the deployment budget is small, a larger $W_{\mathrm{act}}$ means that a smaller number of active elements can be deployed, thus resulting in reduced ergodic capacity.
When $W_{\mathrm{act}}$ exceeds a threshold, all the deployment budget is assigned to passive elements for maximizing the ergodic capacity. As such, increasing $W_{\mathrm{act}}$ will not change the ergodic capacity and elements allocation.
\section{Conclusions}\label{sec_conclu}
In this paper, we proposed a new hybrid active-passive IRS architecture and studied its optimal active/passive elements allocation in a hybrid IRS aided wireless system based on the statistical CSI only.
Specifically, under the general Rician fading channel model, we formulated an optimization problem to maximize the ergodic capacity of the worst-case user, by jointly optimizing the active/passive elements allocation, their phase shifts, and the amplification factors of active elements, subject to various practical constraints on the active-element amplification factor and amplification power consumption, as well as the total active and passive elements deployment budget.
To solve this problem, we first approximated the ergodic capacity in a simpler form and then proposed an efficient algorithm to obtain the optimal hybrid IRS beamforming and active/passive elements allocation.
Moreover, it was shown that when all channels are of LoS, only active elements need to be deployed when the total deployment budget is sufficiently small, while both active and passive elements should be deployed with a decreasing number ratio when the budget increases and exceeds a certain threshold.
Last, numerical results demonstrated the performance gains achieved by the proposed hybrid IRS architecture with the optimal active/passive elements allocation against the benchmarks of the fully-active and fully-passive IRSs, as well as the hybrid IRS with equal active/passive elements allocation. This validated that the hybrid IRS can flexibly balance the trade-off between the peculiar power amplification gain of active IRS and superior beamforming gain of passive IRS.

\appendices
\section{Proof of Lemma \ref{lem_C_approx}}\label{proof_lem1}
First, based on the Lemma 1 of \cite{6816003}, we obtain the following approximation for the ergodic capacity. 
\begin{align}
    C&=\mathbb{E}\left\{\log _{2}\left(1+\frac{P_{\mathrm{B}}|(\mathbf{h}_{\mathrm{IU}}^{\mathrm{act}})^H\mathbf{\Psi}^{\mathrm{act}}\mathbf{h}_{\mathrm{BI}}^{\mathrm{act}}+(\mathbf{h}_{\mathrm{IU}}^{\mathrm{pas}})^H\mathbf{\Psi}^{\mathrm{pas}}\mathbf{h}_{\mathrm{BI}}^{\mathrm{pas}}|^2}{\sigma_{\mathrm{I}}^2\|(\mathbf{h}_{\mathrm{IU}}^{\mathrm{act}})^H\mathbf{\Psi}^{\mathrm{act}}\|^2+\sigma_0^2}\right)\right\} \\
    &\approx\log_2\left(1+\frac{\mathbb{E}\left\{P_{\mathrm{B}}|(\mathbf{h}_{\mathrm{IU}}^{\mathrm{act}})^H\mathbf{\Psi}^{\mathrm{act}}\mathbf{h}_{\mathrm{BI}}^{\mathrm{act}}+(\mathbf{h}_{\mathrm{IU}}^{\mathrm{pas}})^H\mathbf{\Psi}^{\mathrm{pas}}\mathbf{h}_{\mathrm{BI}}^{\mathrm{pas}}|^2\right\}}{\mathbb{E}\left\{\sigma_{\mathrm{I}}^2\|(\mathbf{h}_{\mathrm{IU}}^{\mathrm{act}})^H\mathbf{\Psi}^{\mathrm{act}}\|^2+\sigma_0^2\right\}}\right).\label{C_approx}
\end{align}
Then, we focus on the derivation of the desired signals at the worst-case user, which can be expressed as
\begin{align}
    &\mathbb{E}\left\{P_{\mathrm{B}}|(\mathbf{h}_{\mathrm{IU}}^{\mathrm{act}})^H\mathbf{\Psi}^{\mathrm{act}}\mathbf{h}_{\mathrm{BI}}^{\mathrm{act}}+(\mathbf{h}_{\mathrm{IU}}^{\mathrm{pas}})^H\mathbf{\Psi}^{\mathrm{pas}}\mathbf{h}_{\mathrm{BI}}^{\mathrm{pas}}|^2\right\}\\
    &=\frac{P_{\mathrm{B}}}{(K_1+1)(K_2+1)}\mathbb{E}\Big\{\Big|\underbrace{\sqrt{K_1K_2}(\bar{\mathbf{h}}_{\mathrm{IU}}^{\mathrm{act}})^H\mathbf{\Psi}^{\mathrm{act}}\bar{\mathbf{h}}_{\mathrm{BI}}^{\mathrm{act}}}_{x_{1}}+\underbrace{\sqrt{K_1}(\tilde{\mathbf{h}}_{\mathrm{IU}}^{\mathrm{act}})^H\mathbf{\Psi}^{\mathrm{act}}\bar{\mathbf{h}}_{\mathrm{BI}}^{\mathrm{act}}}_{x_{2}}\nonumber\\
    &+\underbrace{\sqrt{K_2}(\bar{\mathbf{h}}_{\mathrm{IU}}^{\mathrm{act}})^H\mathbf{\Psi}^{\mathrm{act}}\tilde{\mathbf{h}}_{\mathrm{BI}}^{\mathrm{act}}}_{x_{3}}+\underbrace{(\tilde{\mathbf{h}}_{\mathrm{IU}}^{\mathrm{act}})^H\mathbf{\Psi}^{\mathrm{act}}\tilde{\mathbf{h}}_{\mathrm{BI}}^{\mathrm{act}}}_{x_{4}}+\underbrace{\sqrt{K_1K_2}(\bar{\mathbf{h}}_{\mathrm{IU}}^{\mathrm{pas}})^H\mathbf{\Psi}^{\mathrm{pas}}\bar{\mathbf{h}}_{\mathrm{BI}}^{\mathrm{pas}}}_{x_{5}}\\
    &+\underbrace{\sqrt{K_1}(\tilde{\mathbf{h}}_{\mathrm{IU}}^{\mathrm{pas}})^H\mathbf{\Psi}^{\mathrm{pas}}\bar{\mathbf{h}}_{\mathrm{BI}}^{\mathrm{pas}}}_{x_{6}}+\underbrace{\sqrt{K_2}(\bar{\mathbf{h}}_{\mathrm{IU}}^{\mathrm{pas}})^H\mathbf{\Psi}^{\mathrm{pas}}\tilde{\mathbf{h}}_{\mathrm{BI}}^{\mathrm{pas}}}_{x_{7}}+\underbrace{(\tilde{\mathbf{h}}_{\mathrm{IU}}^{\mathrm{pas}})^H\mathbf{\Psi}^{\mathrm{pas}}\tilde{\mathbf{h}}_{\mathrm{BI}}^{\mathrm{pas}}}_{x_{8}}\Big|^2\Big\}\nonumber\\
    &=\frac{P_{\mathrm{B}}}{(K_1\!+\!1)(K_2\!+\!1)}\Big(\mathbb{E}\left\{\left|x_{1}\!+\!x_{5}\right|^2\right\}\!+\!\mathbb{E}\Big\{\left|x_{2}\right|^2\!+\!\left|x_{3}\right|^2\!+\!\left|x_{4}\right|^2\!+\!\left|x_{6}\right|^2\!+\!\left|x_{7}\right|^2\!+\!\left|x_{8}\right|^2\Big\} \Big)\\
    &=x_{\mathrm{L}}+x_{\mathrm{NL,act}}+x_{\mathrm{NL,pas}},\label{lem_C_approx_1}
\end{align}
where $x_{\mathrm{L}}$, $x_{\mathrm{NL,act}}$, and $x_{\mathrm{NL,pas}}$ are defined in \eqref{x_1}, \eqref{x_2}, and \eqref{x_3}, respectively.
Next, the noise introduced by active reflecting elements and that at the receiver can be expressed as
\begin{align}
    &\mathbb{E}\left\{\sigma_{\rm I}^2\left\|({\mathbf{h}}_{\mathrm{IU}}^{\mathrm{act}})^H\mathbf{\Psi}^{\mathrm{act}}\right\|^2+\sigma_0^2\right\}\\
    &=\mathbb{E}\left\{\sigma_{\rm I}^2\left\|\sqrt{\frac{K_2}{K_2+1}}(\bar{\mathbf{h}}_{\mathrm{IU}}^{\mathrm{act}})^H\mathbf{\Psi}^{\mathrm{act}}+\sqrt{\frac{1}{K_2+1}}(\tilde{\mathbf{h}}_{\mathrm{IU}}^{\mathrm{act}})^H\mathbf{\Psi}^{\mathrm{act}}\right\|^2\right\}+\sigma_0^2\\
    &=\underbrace{\frac{K_2\sigma_{\rm I}^2}{K_2+1}\|(\bar{\mathbf{h}}_{\mathrm{IU}}^{\mathrm{act}})^H\mathbf{\Psi}^{\mathrm{act}}\|^2}_{z_{\mathrm{L,act}}}+\underbrace{\frac{\sigma_{\rm I}^2}{K_2+1}\|(\tilde{\mathbf{h}}_{\mathrm{IU}}^{\mathrm{act}})^H\mathbf{\Psi}^{\mathrm{act}}\|^2}_{z_{\mathrm{NL,act}}}+\sigma_0^2.\label{noise_approx}
\end{align}
Last, the proof is completed by substituting \eqref{lem_C_approx_1} and \eqref{noise_approx} into \eqref{C_approx}.
\section{Proof of Lemma \ref{lem1}}\label{proof_lem2}
First, given favorable amplification power condition in \eqref{cons_pb}, it can be verified that at the optimal solution to problem (P3), the constraint \eqref{cstr_power_HI_decomp_2} is always active, i.e., $\sum_{n=1}^{N_{\mathrm{act}}}\alpha^2_{n}=A_{\mathrm{sum}}$. Then, by substituting $A_{\mathrm{sum}}$ into \eqref{x_NLoS} and \eqref{n_act}, we can show that the NLoS part of the desired signal, $x_{\mathrm{NL,act}}+x_{\mathrm{NL,pas}}$, and the amplification noise , $z_{\mathrm{L,act}}+z_{\mathrm{NL,act}}$, are constants. Thus, problem (P3) can be solved by maximizing the ergodic capacity due to the LoS channel component, $x_{\mathrm{L}}$. 
By using the Cauchy-Schwarz inequality, we have
\begin{align}
    x_{\mathrm{L}}&=\gamma_1\left(\sum_{n=1}^{N_{\mathrm{act}}}\alpha_n+N_{\mathrm{pas}}\right)^2P_{\mathrm{B}}\beta^2/D^2_{\mathrm{BI}}d^2_{\mathrm{IU}}\\
    &=\gamma_1\left(\sum_{n=1}^{N_{\mathrm{act}}}\left(\alpha_n+N_{\mathrm{pas}}/N_{\mathrm{act}}\right)\right)^2P_{\mathrm{B}}\beta^2/D^2_{\mathrm{BI}}d^2_{\mathrm{IU}}\\
    &\leq \gamma_1N_{\mathrm{act}}^2\left(\alpha^*+N_{\mathrm{pas}}/N_{\mathrm{act}}\right)^2P_{\mathrm{B}}\beta^2/D^2_{\mathrm{BI}}d^2_{\mathrm{IU}},
\end{align}
where the equality holds if and only if $\alpha_{n}=\alpha^*,\forall n \in\mathcal{N}_{\mathrm{act}}$ with $\alpha^*=\sqrt{\frac{P_{\mathrm{I}}/N_{\mathrm{act}}}{P_{\mathrm{B}}\beta/D^2_{\mathrm{BI}}+\sigma^2_{\mathrm{I}}}}$ in \eqref{opt_a_n}, which thus completes the proof.
\section{Proof of Theorem \ref{the_opt_N}}\label{proof_lem3}
\begin{table}[t]\centering
\caption{Variations of the Receiver SNR under Different Conditions.}\label{opt_N_A}
\begin{center}
\begin{tabular}{|c|c|}
\hline 
Condition & Variations of $\gamma_{\mathrm{hyb}}(\tilde N_{\mathrm{act}})$ \\
\hline
$0<\sqrt{\frac{W_0}{W_{\mathrm{act}}}}<n_{\mathrm{A},2}$ & Increase for $\tilde N_{\mathrm{act}}\in(0,\frac{W_0}{W_{\mathrm{act}}}]$\\
\hline 
$n_{\mathrm{A},2}\leq \sqrt{\frac{W_0}{W_{\mathrm{act}}}}\leq n_{\mathrm{A},3}$ & Increase for  $\tilde N_{\mathrm{act}}\in(0,n^2_{\mathrm{A},2}]$, decrease for $\tilde N_{\mathrm{act}}\in(n^2_{\mathrm{A},2},\frac{W_0}{W_{\mathrm{act}}}]$\\
\hline 
$\sqrt{\frac{W_0}{W_{\mathrm{act}}}}>n_{\mathrm{A},3}$ &\!\! Increase for $\tilde N_{\mathrm{act}}\in(0,n^2_{\mathrm{A},2}]\cup(n^2_{\mathrm{A},3},\frac{W_0}{W_{\mathrm{act}}}]$, decrease for $\tilde N_{\mathrm{act}}\in(n^2_{\mathrm{A},2},n^2_{\mathrm{A},3}]$ \\
\hline 
\end{tabular}
\end{center}
\end{table}
For the receiver SNR, $\gamma_{\mathrm{hyb}}(\tilde N_{\mathrm{act}}) = \xi_1\left(-\tilde N_{\mathrm{act}}+\xi_2\sqrt{\tilde N_{\mathrm{act}}}+\xi_3\right)^2$, its first-order derivative over $\sqrt{\tilde N_{\mathrm{act}}}$ can be expressed as
\begin{equation}
    \frac{\partial \gamma_{\mathrm{hyb}}(\tilde N_{\mathrm{act}})}{\partial \sqrt{\tilde N_{\mathrm{act}}}}=2\xi_1\left(-\tilde N_{\mathrm{act}}+\xi_2 \sqrt{\tilde N_{\mathrm{act}}}+\xi_3\right)(-2  \sqrt{\tilde N_{\mathrm{act}}}+\xi_2).\label{deri_snr}
\end{equation}

First, it can be shown that \eqref{deri_snr} equals to 0 for $\sqrt{\tilde N_{\mathrm{act}}}=n_{\mathrm{A},1}\triangleq\frac{\xi_2-\sqrt{\xi_2^2+4\xi_3}}{2}$, $\sqrt{\tilde N_{\mathrm{act}}}=n_{\mathrm{A},2}=\frac{\xi_2}{2}$, and $\sqrt{\tilde N_{\mathrm{act}}}=n_{\mathrm{A},3}\triangleq\frac{\xi_2+\sqrt{\xi_2^2+4\xi_3}}{2}$ with $n_{\mathrm{A},1}<0<n_{\mathrm{A},2}<n_{\mathrm{A},3}$.
Then, we summarize in Table \ref{opt_N_A} the optimal number of active elements under different conditions of $\tilde N_{\mathrm{act}}$.

Based on Table \ref{opt_N_A}, when $n_{\mathrm{A},2}\leq \sqrt{\frac{W_0}{W_{\mathrm{act}}}}\leq n_{\mathrm{A},3}$, it is obtained that $\gamma_{\mathrm{hyb}}(\tilde N_{\mathrm{act}})$ increases for $0<\tilde N_{\mathrm{act}}\leq n_{\mathrm{A},2}^2$ and decreases for $n_{\mathrm{A},2}^2<\tilde N_{\mathrm{act}}\leq\frac{W_0}{W_{\mathrm{act}}}$, thus leading to the optimal number of active elements, $\tilde N^*_{\mathrm{act}}=n_{\mathrm{A},2}^2$, for maximizing $\gamma_{\mathrm{hyb}}(\tilde N_{\mathrm{act}})$. Following the similar procedure, the optimal number of active elements under different conditions can be obtained.
Note that when $\sqrt{\frac{W_0}{W_{\mathrm{act}}}}>n_{\mathrm{A},3}$, we have
\begin{equation}
    \gamma_{\mathrm{hyb}}(n_{\mathrm{A},2}^2)-\gamma_{\mathrm{hyb}}(\frac{W_0}{W_{\mathrm{act}}})=\frac{\xi_2^2}{4}+\xi_3-\xi_2\sqrt{\xi_3}=(\frac{\xi_2}{2}-\xi_3)^2\geq 0,
\end{equation}
such that $\gamma_{\mathrm{hyb}}(n_{\mathrm{A},2}^2)\geq \gamma_{\mathrm{hyb}}(\frac{W_0}{W_{\mathrm{act}}})$.
Moreover, when $0<\sqrt{\frac{W_0}{W_{\mathrm{act}}}}<n_{\mathrm{A},2}$, i.e., the deployment cost $W_0 \leq \frac{W_{\mathrm{pas}}^{2} P_{\mathrm{I}} / W_{\mathrm{act}}}{4 P_{\mathrm{B}} \beta / D_{\mathrm{BI}}^{2}+4 \sigma_{\mathrm{I}}^{2}}$, all the deployment budget should be assigned to active elements, i.e., $\tilde N^*_{\mathrm{act}}=\frac{W_0}{W_{\mathrm{act}}}$.
Based on the above, the optimal active/passive elements allocation for the hybrid IRS is given by \eqref{opt_na}, thus completing the proof.
\section{Proof of Theorem \ref{lem_thres}}\label{proof_lem4}
First, we compare the achievable capacity by the fully-passive and fully-active IRSs in \eqref{C_p} and \eqref{C_a}, respectively. By solving
\begin{equation}
    \frac{W_0A_{\mathrm{sum}} P_{\mathrm{B}} \beta^{2} /W_{\mathrm{act}} D_{\mathrm{BI}}^{2} d_{\mathrm{IU}}^{2}}{A_{\mathrm{sum}} \sigma_{\mathrm{I}}^{2} \beta / d_{\mathrm{IU}}^{2}+\sigma_{0}^{2}}>\frac{W_0^2P_{\mathrm{B}}\beta^2}{W_{\mathrm{pas}}^2D_{\mathrm{BI}}^2d_{\mathrm{IU}}^2\sigma_0^2},
\end{equation}
it is obtained that when
\begin{equation}
    W_0 > W_{\mathrm{A-P}}\triangleq\frac{W_{\mathrm{pas}}^2/W_{\mathrm{act}}}{\frac{\beta\sigma_{\mathrm{I}}^2}{d_{\mathrm{IU}}^2\sigma_0^2}+\frac{P_{\mathrm{B}} \beta / D_{\mathrm{BI}}^{2}+\sigma_{\mathrm{I}}^{2}}{P_{\mathrm{I}}}},
\end{equation}
the fully-passive IRS outperforms the fully-active IRS in term of the capacity.
Second, it is obtained from \eqref{deri_snr} that when $W_0>W_{\mathrm{A-H}}$, the hybrid IRS outperforms the fully-active IRS, and achieves its maximum.
Third, we compare the achievable capacity by the hybrid IRS and passive IRS. By comparing \eqref{C_p} and \eqref{C_h_2}, it can be shown that when
\begin{equation}
    \frac{W_0^2P_{\mathrm{B}}\beta^2}{W_{\mathrm{pas}}^2D_{\mathrm{BI}}^2d_{\mathrm{IU}}^2\sigma_0^2}>\frac{P_{\mathrm{B}}\beta^2\left(\frac{A_{\mathrm{sum}}W_{\mathrm{pas}}}{4W_{\mathrm{act}}}+\frac{W_0}{W_{\mathrm{pas}}}\right)^2/D_{\mathrm{BI}}^2d_{\mathrm{IU}}^2}{A_{\mathrm{sum}}\sigma_{\mathrm{I}}^2\beta/d_{\mathrm{IU}}^2+\sigma_0^2},\label{con_hp}
\end{equation}
the fully-passive IRS achieves the maximum capacity.
The condition in \eqref{con_hp} can be re-expressed as
\begin{equation}\label{ineq_H-P}
    \frac{\sigma_{\mathrm{I}}^2\beta}{d_{\mathrm{IU}}^2}W_0^2-\frac{W_{\mathrm{pas}}^2\sigma_0^2}{2W_{\mathrm{act}}}W_0-\frac{A_{\mathrm{sum}}W_{\mathrm{pas}}^4\sigma_0^2}{16W_{\mathrm{act}}^2}>0,
\end{equation}
which leads to
\begin{equation}
    W_0<\frac{W_{\mathrm{pas}}^2\sigma_0^2d_{\mathrm{IU}}^2}{4W_{\mathrm{act}}\sigma_{\mathrm{I}}^2\beta}-\frac{W_{\mathrm{pas}}^2\sigma_0d_{\mathrm{IU}}}{4W_{\mathrm{act}}\sigma_{\mathrm{I}}}\sqrt{\frac{\sigma_0^2d_{\mathrm{IU}}^2}{\sigma_{\mathrm{I}}^2\beta^2}+\frac{P_{\mathrm{I}}}{P_{\mathrm{B}}\beta^2/D_{\mathrm{BI}}^2+\sigma_{\mathrm{I}}^2\beta}}<0,
\end{equation}
and $W_0>W_{\mathrm{H-P}}$. 
Last, we discuss the relations of $W_{\mathrm{A-P}},W_{\mathrm{A-H}}$ and $W_{\mathrm{H-P}}$
to facilitate our analysis, we list all possible permutations of the above thresholds in Table \ref{per_thres}.
\begin{table}[t]\centering
\caption{Possible relations for budget thresholds.}\label{per_thres}
\begin{tabular}{|c|c|}
\hline
\begin{tabular}[c]{@{}c@{}}Permutations\end{tabular} & \begin{tabular}[c]{@{}c@{}}$W_{\mathrm{A-H}}<W_{\mathrm{A-P}}<W_{\mathrm{H-P}},W_{\mathrm{A-H}}<W_{\mathrm{H-P}}<W_{\mathrm{A-P}},W_{\mathrm{A-P}}<W_{\mathrm{A-H}}<W_{\mathrm{H-P}},$\\ $W_{\mathrm{H-P}}<W_{\mathrm{A-P}}<W_{\mathrm{A-H}},W_{\mathrm{A-P}}<W_{\mathrm{H-P}}<W_{\mathrm{A-H}},W_{\mathrm{H-P}}<W_{\mathrm{A-H}}<W_{\mathrm{A-P}}.$\end{tabular} \\ \hline
\end{tabular}
\end{table}
\noindent It can be shown that only one permutation of the above thresholds is feasible, i.e., $W_{\mathrm{A-H}}<W_{\mathrm{A-P}}<W_{\mathrm{H-P}}$. 
Specifically, the hybrid IRS reduces to the fully-active IRS when $W_0<W_{\mathrm{A-H}}$; it reduces to the fully-passive IRS when $W_0>W_{\mathrm{H-P}}$; and it outperforms the fully-active and fully-passive IRSs with the optimal number of active elements of $\tilde N_{\mathrm{act}}=\tilde N^*_{\mathrm{act}}$ otherwise. 
Moreover, all other permutations can be verified to be infeasible by contradiction. For example, for the permutation $W_{\mathrm{A-H}}<W_{\mathrm{H-P}}<W_{\mathrm{A-P}}$, when $W_{\mathrm{A-H}}<W_{\mathrm{H-P}}<W_0<W_{\mathrm{A-P}}$, we obtain that the hybrid IRS outperforms the fully-active IRS when $W_0>W_{\mathrm{A-H}}$, the fully-passive IRS outperforms hybrid IRS when $W_0>W_{\mathrm{H-P}}$, and the fully-active IRS outperforms the fully-passive IRS when $W_0<W_{\mathrm{A-P}}$, which contradicts each other and thus is infeasible.
Combining the above leads to the desired result.

\end{document}